\documentclass[sigconf, nonacm]{acmart}

\usepackage{array}
\usepackage{booktabs}      
\usepackage{pgf}
\usepackage{tabularx}      
\usepackage{subcaption}    
\usepackage{graphicx}      
\usepackage{cleveref}
\usepackage{algorithm}
\usepackage{algpseudocode}
\usepackage{multirow}
\usepackage{tikz}
\usetikzlibrary{arrows.meta,positioning,fit,calc,decorations.pathreplacing, patterns, backgrounds}
\usepackage{xcolor}
\usepackage{pifont}  
\usepackage{listings}
\usepackage{setspace}
\usepackage{subcaption}
\usepackage{comment}

\definecolor{nousblue}{HTML}{2E5EAA}
\definecolor{nousgray}{HTML}{6B7280}
\definecolor{nousgreen}{HTML}{059669}
\definecolor{nousred}{HTML}{DC2626}
\definecolor{nousorange}{HTML}{D97706}
\definecolor{nouspurple}{HTML}{7C3AED}

\newcommand{\squishlist}{
 \begin{list}{$\bullet$}
  { \setlength{\itemsep}{0pt}
     \setlength{\parsep}{3pt}
     \setlength{\topsep}{3pt}
     \setlength{\partopsep}{0pt}
     \setlength{\leftmargin}{1.5em}
     \setlength{\labelwidth}{1em}
     \setlength{\labelsep}{0.5em} } }

\newcommand{\squishend}{
  \end{list}  }

\newcommand{\taumax}{\tau_{\max}}

\newcommand{\Cmax}{C_{\max}}

\newcommand{\sch}{\mathrm{sch}}

\definecolor{Georgia}{RGB}{112, 95, 214}
\definecolor{Chris}{RGB}{50, 168, 90}

\newtheorem{property}{\small Property}[subsection]

\providecommand{\attr}[1]{\mathbf{#1}}                     
\providecommand{\sch}{\mathrm{sch}}                        
\providecommand{\ext}{\varepsilon}                         
\providecommand{\semjoin}{\mathbin{\Join^{\mathrm{sem}}}}  

\providecommand{\roleR}{\textbf{Rewrite}}    
\providecommand{\roleX}{\textbf{Exec.}}      
\providecommand{\roleC}{\textbf{Cost}}       
\providecommand{\roleF}{\textit{Future}}     

\begin{document}
\title{NousDB - A Fast Query Engine for Semantic Analytics}





\author{Christos Tsapelas}
\affiliation{%
  \institution{Department of
Informatics and Telecommunications, National and Kapodistrian University
of Athens \newline Archimedes, Athena Research Center}
   \city{Athens}
  \country{Greece}
}
\email{ctsapelas@athenarc.gr}
\email{ctsapelas@di.uoa.gr}

\author{Georgia Koutrika}
\affiliation{%
  \institution{Archimedes, Athena Research Center}
   \city{Athens}
  \country{Greece}
}
\email{georgia@athenarc.gr}

\author{Yannis Ioannidis}
\affiliation{%
  \institution{Department of
Informatics and Telecommunications, National and Kapodistrian University
of Athens \newline Archimedes, Athena Research Center}
 \city{Athens}
  \country{Greece}
}
\email{yannis@athenarc.gr}
\email{yannis@di.uoa.gr}

\begin{abstract}
Semantic data analytics systems integrate AI-powered operators
into query pipelines, enabling declarative queries over structured
and unstructured data. 
However, existing systems rely on per-tuple model invocations, ad hoc operator designs, and lack principled cost models, resulting in poor scalability and limited optimization.
We present NousDB, a high-performance semantic analytics engine built on Apache DataFusion\cite{10.1145/3626246.3653368}.Through NousDB we propose an extended relational algebra with three semantic operators and formal algebraic properties that enable principled query rewriting and cross-operator optimization. 
We further propose a token-based cost model using equi-depth token histograms to optimize batch sizing under context-window constraints and provide statistical guarantees for safe batching. 
Experiments show that NousDB achieves up to $\times 86$ speedups with both local and enterprise LLMs, while leading competitors in token efficiency and accuracy. These gains extend to complex JOB and Stack workloads augmented with semantic operators, supporting both traditional and semantic analytics.

\end{abstract}

\maketitle


\section{Introduction} Foundation models~\cite{openai2024gpt4technicalreport, touvron2023llama, anthropic2024claude} have showcased strong capabilities in processing unstructured data, such as language and images.
This new form of data processing has given rise to semantic query engines (SQEs), which treat models as query operators and allow users to filter, join, extract information from, and summarize textual content in combination with conventional relational operations.
However, they challenge
several fundamental assumptions underlying database systems, including operator semantics, cost estimation, and parallel execution, among others.
 
\noindent\textbf{Motivating example.}
Consider the query ``which actors are mentioned in movie reviews?'' over the IMDB schema extended with a relation of reviews. 
\Cref{fig:running-example} shows this query, denoted Q11 in our evaluation, expressed in NousDB. Its natural execution plan first extends each review with the actors mentioned in its text, join the result with the \textsf{name} relation on whether the extracted mention and a person's name refer to the same actor, and projects the review, movie title and actor name.
A naive execution invokes the model once per review for extraction and once per review-name pair for matching, with each invocation processing only one tuple or tuple pair.
Avoiding this costly tuple-at-a-time execution requires the planner to reason about semantic operators, for which existing SQEs provide insufficient foundations.
Using this query as a running example, we next identify the following problems.


\definecolor{kw}{HTML}{1A7F37}       
\definecolor{str}{HTML}{CF222E}      
\definecolor{cmt}{HTML}{6E7781}      
\definecolor{fn}{HTML}{8250DF}       
\definecolor{tp}{HTML}{0550AE}       
\definecolor{bg}{HTML}{F6F8FA}       
\definecolor{frame}{HTML}{D0D7DE}    

\lstdefinelanguage{NousDBRust}{
  morekeywords=[1]{let,fn,pub,async,await,use,mod,struct,impl,trait,where,
                   for,in,if,else,match,return,mut,ref,move,dyn,Arc},
  morekeywords=[2]{extend,semantic_join,select_columns,nousdb_sql,collect,show},
  morekeywords=[3]{Utf8,DataType,NousDBContext,NousDBDataFrame},
  sensitive=true,
  morecomment=[l]{//},
  morecomment=[s]{/*}{*/},
  morestring=[b]",
}
\lstdefinestyle{nousdb}{
  language=NousDBRust,
  basicstyle=\ttfamily\scriptsize,
  keywordstyle=[1]\bfseries\color{kw},
  keywordstyle=[2]\bfseries\color{fn},
  keywordstyle=[3]\color{tp},
  commentstyle=\itshape\color{cmt},
  stringstyle=\color{str},
  backgroundcolor=\color{bg},
  frame=tb,
  rulecolor=\color{frame},
  framesep=2pt,
  xleftmargin=3pt,
  xrightmargin=3pt,
  numbers=left,
  numberstyle=\tiny\color{cmt},
  numbersep=2pt,
  breaklines=true,
  columns=fullflexible,
  keepspaces=true,
  showstringspaces=false,
  tabsize=2,
  aboveskip=2pt,
  belowskip=2pt,
  lineskip=-0.5pt,
  literate={->}{{$\to$}}1 {=>}{{$\Rightarrow$}}1,
}

\begin{figure}[t]
\centering
\begin{lstlisting}[style=nousdb]
// Q11: which actors are mentioned in movie reviews?
let reviews = ctx.nousdb_sql("SELECT * FROM reviews")
  .extend(("mentioned_actor", Utf8), model,
          "Extract the actor name mentioned in this review", &["review"]);
let names = ctx.nousdb_sql("SELECT * FROM name");
let result = reviews
  .semantic_join(names, &["mentioned_actor"], &["name"], model,
                 "Does the extracted actor match this person?")
  .select_columns(&["review", "title", "name"]);
result.show().await?;
\end{lstlisting}
\vspace{-14pt}
\caption{\footnotesize NousDB example program.}
\vspace{-16pt}
\label{fig:running-example}
\end{figure}

\noindent\textbf{Problem 1: Semantic operators without semantics.}
Existing SQEs expose semantic operators to the query planner as opaque model invocations that use task-specific prompts~\cite{lotus,palimpzest,blendsql}.
This limitation hinders query optimization in SQEs because it prevents the definition of algebraic properties to reason about result-preserving rewrites.
In our example, treating semantic operators as black boxes prevents the planner from determining whether it can safely push a selective filter on \textsf{name} below the semantic join, fuse mention extraction and matching into a single model invocation, or partition the work across threads.
Missing such rewrites can incur substantial model work, often linear in the input cardinality and far costlier per invocation than conventional tuple comparisons.

\textit{Challenge 1: Defining semantics that enable sound plan transformations}. A semantic operator is jointly determined by its model and prompt, may behave stochastically, and may be undefined for some inputs.
Classical relational-algebra properties therefore cannot be assumed. The challenge is to define semantics that make these behaviors explicit and identify which properties, and thus which plan transformations, are valid under which assumptions. 

\noindent\textbf{Problem 2: Finite context windows.} Singleton calls underutilize the context window.
 In our running example, invoking the model separately for every review and review--name pair repeatedly pays the fixed cost of the prompt while leaving most of the context window unused. Batching multiple tuples into each invocation amortizes this cost and reduces the number of model calls,  but risks exceeding the model's capacity.  Thus, batching must respect a hard token budget determined by the prompt and the model's context window, and  
  becomes a critical optimization decision. 

\textit{Challenge 2: Token-distribution estimation.} 
Enforcing this budget requires estimating the token footprint of the input tuples without tokenizing the entire input in advance. The engine therefore needs compact per-column token statistics that can be collected and maintained at negligible time and space overhead. 

\textit{Challenge 3: Probabilistic batch sizing.} Given these statistics and the available token budget, the engine must select the largest batch that fits with high probability. The sizing algorithm must balance invocation efficiency against the risk of truncation or rejection, particularly under skewed token-length distributions.

 \noindent\textbf{Problem 3: Over-batching can degrade result quality.} 
 Although a batch may fit within the model’s context window, batching can change the output produced for each tuple relative to singleton execution due to  leakage from neighboring tuples, tuple position in the batch, or formatting errors that misalign inputs and outputs.   In our experiments, batching without checking output fidelity reduces F1 by up to 0.42 for a selection query with Qwen3-8B. The engine must therefore determine the largest batch size that preserves fidelity, meaning agreement with singleton execution, rather than merely the largest batch that fits within the context window.


\textit{Challenge 4: Safe batch sizes cannot be predicted a priori.} The batch size at which output quality degrades depends on the model, prompt, and data. At query time, fidelity must be assessed without ground truth or independence assumptions and at a cost substantially lower than the execution being validated.
\noindent\textbf{Problem 3: Token-skew in row-balanced partitions.}
Classical engines distribute rows evenly across workers, assuming similar per-row costs. 
 In SQEs, cost depends on token volume, so row-balanced partitions may have highly unequal loads. In our example, a partition containing longer reviews carries several times the token load of its peers, becoming a straggler while other workers sit idle. We call this imbalance \emph{token skew}. Existing SQEs inherit row-based partitioning from their host engines and do not account for it.

 
 \textit{Challenge 5: Token-aware partitioning.}  The challenge is to detect token skew and partition tuples by expected processing cost rather than count, while keeping repartitioning overhead below the gains from improved load balance.

\textbf{NousDB Contributions.}
 NousDB addresses all four problems. It defines semantic operators declaratively as relation-to-relation functions parameterized by a model and a prompt.
 These definitions provide engine-independent reference semantics for evaluating implementations and reasoning about plan equivalence.
 We define these semantics at the tuple level, separating \emph{what} an operator computes from \emph{how} its inputs are grouped into model invocations. 
 Singleton execution, in which each tuple is evaluated independently, serves as the reference semantics. 
 Batching is therefore a physical execution choice whose fidelity can be evaluated against this reference, rather than an implicit change to the operator's meaning. 
  NousDB is built as an extension of Apache DataFusion~\cite{10.1145/3626246.3653368} over Apache Arrow's~\cite{apache_arrow} columnar format (\Cref{fig:system-architecture}). 
  Its semantic operators are first-class DataFusion operators, enabling vectorized, columnar, and partition-parallel execution.  
In a nutshell, NousDB makes the following contributions:
\squishlist

\item \textbf{An algebra for semantic operators.} We define three tuple-level semantic operators (Extend, Selection, Join) under an explicit oracle assumption, which allows us to derive a set of algebraic properties, along with four rewrite rules, two pushdowns and two fusions, that NousDB's optimizer applies.

\item \textbf{Token-aware batch sizing.}
We introduce equi-depth token histograms as the catalog statistic of token-priced operators, an iterative algorithm that maximizes $k$ subject to an overflow probability $\alpha$, and batched streaming operators that keep a window of packed requests in flight.

\item \textbf{A fidelity certificate.} We define \emph{batch fidelity} as the probability that a tuple has the same output in a batch as in a singleton invocation and certify it at query time with a one-probe-per-batch calibration that restores independent samples and an exact Clopper--Pearson test.
In query execution, certified runs agree with singleton execution on at least 86\% of the tuples of every query.

\item \textbf{Token-aware repartitioning.}
We implement a token repartition operator with list-scheduling and zigzag assignment, a corrected balance guarantee, and a wave-barrier protocol that embeds it in DataFusion execution engine.

\item \textbf{Strong performance.} On IMDB and SemBench, against LOTUS, Palimpzest, Abacus and per-row DataFusion, with a local (Qwen3-8B) and a hosted (GPT-5.4-nano) model, NousDB is the fastest system, up to $86\times$, uses up to $19\times$ fewer tokens, and lies on the quality--runtime Pareto frontier.
On JOB and STACK augmented with semantic operators (43 queries, up to 17 relations) it keeps its lead, up to $35\times$.
\squishend


\section{Related Work} 

SQEs expose LLM-powered operators next to relational ones \cite{thalamusdb1, wannadb1, andb, aryn, agenticdata, docdb, docetl, doctopus, eleet1, exasol, lotus, palimpzest, swelldb, unify, tabagent, weaver, scaledoc, zendb, blendsql}.
ThalamusDB \cite{thalamusdb1, thalamusdb2} answers SQL queries with NL predicates over multi-modal data, WannaDB \cite{wannadb1, wannadb2} and ELEET \cite{eleet1, eleet2} extract structure from text, and ZenDB \cite{zendb} and DocETL \cite{docetl} target document analytics.
Palimpzest \cite{palimpzest}, optimized by Abacus \cite{abacus}, and LOTUS \cite{lotus, tag} offer declarative semantic operators, UQE \cite{uqe} and BlendSQL \cite{blendsql} embed them in SQL, and GALOIS \cite{Galois} and XMODE \cite{xmode} query or orchestrate the LLM itself.
Benchmarks remain few \cite{tag, deepscholar, swan}; the broadest, SemBench \cite{sembench}, shows that operators sharing a name across systems, such as \textsf{sem\_filter} and \textsf{sem\_join}, do not share semantics.
In all of them a semantic operator is an opaque model call.
NousDB instead defines semantic operators that compose with the relational ones, establishes the algebraic properties that rewriting requires, and measures cost in tokens rather than rows, which exposes \emph{batch sizing}, \emph{batch fidelity} and \emph{repartitioning} as new optimization problems.

\noindent\textbf{Batching and guarantees.}
Batch prompting amortizes instruction tokens over the inputs of one prompt \cite{Cheng2023BatchPE, lin2024batchprompt}, at a quality cost reported as discounted output \cite{Liu2024CliqueParcelAA} and cross-query interference \cite{Xu2026TowardsCL}.
Statistical guarantees for LLM-powered operators \cite{approximate_selection_with_guarantees, 10.1145/3769776, lotus} bound a cheap proxy against an oracle to decide \emph{which} tuples the expensive model must see.
NousDB's certificate is orthogonal: it bounds the fidelity of batched execution to singleton execution (\Cref{def:invariance}) to decide \emph{how many} tuples may share one call.

\section{NousDB Algebra} 

 \definecolor{layerAPI}{HTML}{DCEEFB}
\definecolor{layerPlanner}{HTML}{FDE8D0}
\definecolor{layerOptimizer}{HTML}{D5F5E3}
\definecolor{layerExecutor}{HTML}{F5D5E8}
\definecolor{layerFoundation}{HTML}{E8E8E8}
\definecolor{compFill}{HTML}{FFFFFF}
\definecolor{borderBlue}{HTML}{2E86C1}
\definecolor{borderOrange}{HTML}{E67E22}
\definecolor{borderGreen}{HTML}{27AE60}
\definecolor{borderPink}{HTML}{C0392B}
\definecolor{borderGray}{HTML}{7F8C8D}
\definecolor{nousAccent}{HTML}{2E86C1}
\definecolor{dfAccent}{HTML}{0D0313}

\begin{figure}
\centering
\resizebox{\columnwidth}{90pt}{%
  \tikzset{
    comp/.style={
      draw=#1, fill=compFill, rounded corners=2pt, thick,
      font=\strut, minimum height=0.6cm, inner sep=4pt
    },
    compthin/.style={
      draw=#1, fill=compFill, rounded corners=2pt,
      font=\strut, minimum height=0.5cm, inner sep=3pt
    },
    layerlabel/.style={
      font=\bfseries, anchor=north west
    },
    ann/.style={
      font=\footnotesize\itshape, color=gray!70!black
    },
  }
  \begin{tikzpicture}[node distance=0.2cm]

  \def\totalW{12.5cm}
  \def\layergap{0.3cm}
  \def\narrowW{5cm}
  \def\wideW{7.2cm}
  \def\rowHeight{2cm}

  \node[rounded corners=3pt, fill=layerAPI, draw=borderBlue, thick,
        minimum width=\narrowW, minimum height=1.7cm, anchor=north west]
    (api) at (0,0) {};
  \node[layerlabel, color=borderBlue] at ([shift={(2pt,-2pt)}] api.north west)
    {(1) User-Facing API};
  \node[comp=borderBlue] at ([yshift=-0.15cm] api.center)
    (ctx) {\texttt{NousDBContext}};
  \node[ann, below=0pt of ctx] {wraps \texttt{SessionContext}};

  \node[rounded corners=3pt, fill=layerPlanner, draw=borderOrange, thick,
        minimum width=\wideW, minimum height=1.7cm, anchor=north west]
    (planner) at ([xshift=\layergap] api.north east) {};
  \node[layerlabel, color=borderOrange] at ([shift={(2pt,-2pt)}] planner.north west)
    {(2) Planner};
  \node[comp=borderOrange]
    at ([shift={(-1.85cm,-0.15cm)}] planner.center)
    (qp) {\texttt{NousDBQueryPlanner}};
  \node[comp=borderOrange, right=0.1cm of qp]
    (pp) {\texttt{NousDBPhysicalPlanner}};
  \node[ann, below=0pt of qp] {DataFusion};
  \node[ann, below=0pt of pp] {extension};

  \draw[-Stealth, thick, color=gray!80] (api.east) -- (planner.west);

  \node[rounded corners=3pt, fill=layerOptimizer, draw=borderGreen, thick,
        minimum width=\wideW, minimum height=\rowHeight, anchor=north west]
    (optimizer) at ([yshift=-0.1cm] api.south west) {};
  \node[layerlabel, color=borderGreen] at ([shift={(2pt,-2pt)}] optimizer.north west)
    {(3) Optimizer (4 rules)};

  \node[comp=borderGreen, minimum width=3.1cm, font=\small\strut, anchor=east]
    at ([shift={(-0.08cm,-0.95cm)}] optimizer.north)
    (r1) {\ding{192}\;\textsc{SemRewriting}};
  \node[comp=borderGreen, minimum width=3.1cm, font=\small\strut, anchor=west]
    at ([shift={(0.08cm,-0.95cm)}] optimizer.north)
    (r2) {\ding{193}\;\textsc{TokenRepart}};
  \node[comp=borderGreen, minimum width=3.1cm, font=\small\strut, below=0.08cm of r1]
    (r3) {\ding{194}\;\textsc{HistInjection}};
  \node[comp=borderGreen, minimum width=3.1cm, font=\small\strut, below=0.08cm of r2]
    (r4) {\ding{195}\;\textsc{BatchSizeOpt}};

  \node[rounded corners=3pt, fill=layerExecutor, draw=borderPink, thick,
        minimum width=\narrowW, minimum height=\rowHeight, anchor=north west]
    (executor) at ([xshift=\layergap] optimizer.north east) {};
  \node[layerlabel, color=borderPink] at ([shift={(2pt,-2pt)}] executor.north west)
    {(4) Executor};

  \node[compthin=borderPink, minimum width=2.1cm, anchor=east]
    at ([shift={(-0.08cm,-0.95cm)}] executor.north)
    (op1) {$\varepsilon$\;\textbf{Extend}};
  \node[compthin=borderPink, minimum width=2.1cm, anchor=west]
    at ([shift={(0.08cm,-0.95cm)}] executor.north)
    (op2) {$\sigma^{\text{s}}$\;\textbf{Select}};
  \node[compthin=borderPink, minimum width=2.1cm, below=0.08cm of op1]
    (op3) {$\Join^{\text{s}}$\;\textbf{Join}};
  \node[compthin=borderPink, minimum width=2.1cm, below=0.08cm of op2]
    (op4) {\textbf{TokenRepart}};

  \draw[-Stealth, thick, color=gray!80] (optimizer.east) -- (executor.west);

  \node[rounded corners=3pt, fill=layerFoundation, draw=borderGray, thick,
        minimum width=\totalW, minimum height=0.6cm, anchor=north west]
    (foundation) at ([yshift=-0.1cm] optimizer.south west) {};
  \node[font=\bfseries, color=borderGray, anchor=center]
    at (foundation.center)
    {Apache DataFusion\;\;$\cdot$\;\;Arrow};

  \begin{scope}[on background layer]
    \draw[decorate, decoration={brace, amplitude=6pt, mirror},
          ultra thick, nousAccent]
      ([xshift=-10pt] api.north west) -- ([xshift=-10pt] optimizer.south west)
      node[midway, left=9pt, font=\bfseries, color=nousAccent,
           rotate=90, anchor=south]
      {NousDB};

    \draw[decorate, decoration={brace, amplitude=5pt, mirror},
          thick, dfAccent]
      ([xshift=-10pt] foundation.north west) -- ([xshift=-10pt] foundation.south west)
      node[midway, left=8pt, font=\bfseries, color=dfAccent,
           rotate=90, anchor=south]
      {DataFusion};
  \end{scope}

  \end{tikzpicture}%
}
\vspace{-16pt}
\caption{\small NousDB system architecture.}
\label{fig:system-architecture}
\vspace{-10pt}
\end{figure}
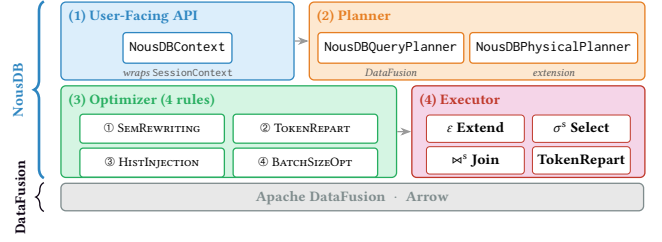

\subsection{NousDB Model}
In the relational model, a relation $R$ has a set $\mathbf{A}$ of attributes, each $A_i\in\mathbf{A}$ with a domain $Dom(A_i)$ (numeric, boolean or categorical); for a tuple $r$, $r[A_i]$ is the value of $A_i$ in $r$, and for $\mathbf{B}\subseteq\mathbf{A}$, $r[\mathbf{B}]$ denotes the projection of $r$ on $\mathbf{B}$.
Relational systems also admit unstructured attributes through general-purpose types such as \texttt{TEXT} or \texttt{BLOB}, but these remain opaque to query semantics: no relational operator or optimizer can reason about their content.
NousDB lifts this limitation by drawing models from a registered pool $\mathcal{M}$ of foundation models (LLMs, VLMs) that interpret and execute natural-language instructions over their inputs.

A model $m\in\mathcal{M}$ maps input values to output values, $m:\prod_{i} Dom(X_i)\to\prod_{j} Dom(Y_j)$, and is steered by a prompt. 
We define a \emph{prompt constructor} \(pc\) from an operator-specific instruction template \(pt\). 
Given an input $x$, \(pc\) instantiates \(pt\) with the relevant attribute values from \(x\), producing the prompt \(p=pc(x)\) submitted to \(m\).
Combining a prompt constructor with the model gives a \emph{prompted model invocation} (simply  model invocation hereafter) through which our semantic operators compose model inference with classical query processing. 

\textbf{Singleton semantics.} We define the reference semantics of semantic operators in terms of \emph{singleton model invocations}: each input tuple, or tuple pair for a join, is evaluated independently.

\textbf{Deterministic-oracle assumption.} Repeated model invocations with the same prompt and input may produce different outputs. We abstract away this stochasticity by assuming that, for a fixed model \(m\), prompt \(p\), and input \(x\), every invocation of \(m(p,x)\) returns the same output. Removing model stochasticity from the logical semantics  allows us to reason about semantic operators using standard set-theoretic arguments. It is analogous to simplifying assumptions used by relational optimizers, such as uniformity and statistical independence, which enable selectivity estimates to compose across predicates~\cite{10.14778/2850583.2850594}.
\footnote{Implementations may reduce output variability through deterministic inference configurations, such as model temperature settings and constrained decoding strategies, but the assumption remains an abstraction used to define the logical semantics.}

The logical semantics impose these requirements: 
\squishlist 
\item \emph{Model totality.} The invocation \(m(p,x)\) is defined and produces a well-formed output for every input \(x\) in the operator's domain. 
\item \emph{Tuple-local prompt construction.} Under singleton execution, the prompt constructor is deterministic and uses only the declared input attributes of the tuple \(r\). Thus, for a fixed template and tuple, it always produces the same prompt. 
\squishend

Together, singleton execution, deterministic-oracle evaluation, and tuple-local prompt construction provide \emph{input locality}: the output associated with a tuple depends only on its declared input attributes and not on other tuples, the container relation, its evaluation order, or its assigned partition. Consequently, each NousDB semantic operator, \textit{Extend}, \textit{Semantic Selection}, and \textit{Semantic Join} that are presented next, defines a well-defined function from relations to relations. Furthermore, NousDB logical semantics allows us to define algebraic properties and rewrite rules.

\subsection{NousDB Query Operators}
\label{sem_ops}


\subsubsection{Extend Operator} \label{definition:extend}

The Extend operator, $\varepsilon_{\mathbf{E}, \mathbf{B}, m, p}(R)$, extends a relation $R$ with a set $\mathbf{E}$ of  attributes that are computed by a model $m$ over \textit{R}'s attributes, given a task instruction prompt $p$.

\begin{definition}[Extend Operator] \label{extend_op}
Let \textit{R} be an input relation, $\mathbf{B}$ a subset  of $R$'s attributes, $\mathbf{E}=\{E_1,\dots,E_k\}$  the desired \emph{extended attributes}, $m$ a model and $p$ the final prompt to the model.

The \emph{Extend operator}, $\varepsilon_{\mathbf{E}, \mathbf{B}, m, p}(R)$, is a relation defined as:
\[
\varepsilon_{\mathbf{E}, \mathbf{B}, m, p}(R) = \left\{ (r, e_1 \dots, e_k) \,\middle|\, \forall r \in R,\;  (e_1, \dots, e_k) = m(p, r[\mathbf{B}]) \right\}
\]

\end{definition}


\subsubsection{Semantic Selection} \label{definition:selection}
In traditional relational algebra, the selection operator $\sigma_\theta(R)$ returns all tuples in a relation $R$ that satisfy a boolean condition $\theta$ defined over $R$'s attributes. 
However, $\theta$ is restricted to comparisons between attribute values  (e.g., $=$, $<$, $\neq$).
NousDB extends the expressive power of selection by allowing a model $m \in \mathcal{M}$ to evaluate the filtering condition.

\begin{definition}[Semantic Selection] \label{selection_op}
Let \textit{R} be an input relation, and $\mathbf{B}$ a subset  of $R$'s attributes.
  Let $mp_m$ be a \emph{model-based selection predicate} defined as a filtering condition in a formal or natural language description, evaluated by a model $m$.
  $mp_m$ is defined as:
    \[
        mp_m \;:\; \operatorname{\textit{Dom}}(\mathbf{X}) \;\to\; \{{true},{false}\}, \quad mp_m(x) = m(p,x), \quad x \in \mathbf{X}
    \]
  
%
  The \emph{semantic selection} of $R$ on attributes $\mathbf{B}$, given a model-based selection predicate $mp_m$, is the relation:
  \[
  \sigma^{\text{sem}}_{\mathbf{B},\,m,\,p}(R)
  \;=\;
  \bigl\{\, r \in R
    \;\big|\;
    mp_m\bigl(r[\mathbf{B}]\bigr) = true
  \,\bigr\}.
\]

\end{definition}

\subsubsection{Semantic Join Operator} \label{definition:join}
The join operator in relational algebra combines tuples from two relations when they satisfy a join condition $\theta$.  
This join condition $\theta$ defines the criteria under which tuples from the input relations are matched, involving comparisons in the form of equality, inequality, or ordering relations between the join attributes of each relation.
NousDB redefines join semantics by replacing traditional predicate evaluation with model-based inference. In NousDB, a join condition is evaluated by a model $m \in \mathcal{M}$, enabling joins to capture semantic relationships across both structured and unstructured data.


\begin{definition}[Semantic Join] \label{join_op}
Let $R$ and $S$ be two input relations.
Let $m$ be a model with prompt $p$. The \textit{model-based join predicate} $\theta_m$ is defined as:
\(
\theta_m(x, y):\operatorname{\textit{Dom}}(\mathbf{X})\times\operatorname{\textit{Dom}}(\mathbf{Y}) \rightarrow \{{true},{false}\}, \)
\[\theta_m(x, y) = m(p, x, y), x \in \mathbf{X}, y \in \mathbf{Y}
\]
Let $\mathbf{B}$ and $\mathbf{C}$ be the join attribute sets from $R$ and $S$ respectively, specifying which attributes from each relation participate in the join condition. 
The semantic join of $R$ and $S$ on attribute sets $\mathbf{B}$ and $\mathbf{C}$ given a model-based join condition $\theta_m$ is defined as:
\[
        R \Join^{sem}_{\mathbf{B},\mathbf{C},m, p} S = 
            \bigl\{(r,s)\,\bigl|\;
            \forall r \in R,\;\forall s \in S, \\[-2pt]
             \theta_m \bigl( r[\mathbf{B}], \; s[\mathbf{C}] \bigr) = {true}\bigr\}.
\]


\end{definition}

The model-based join condition $\theta_m$ can be evaluated by any model $m$ that decides whether two tuples should be joined. 
In all cases, the model returns a boolean outcome used to evaluate $\theta_m$.




\vspace{5pt}

\noindent\textbf{Semantic Operator Properties}.\label{definition:properties}  Under the NousDB model, we can define a set of algebraic properties, emerging from the definitions of the semantic operators.
Table~\ref{tab:properties} summarizes all eighteen properties and marks which ones the current implementation NousDB supports. 
Three (E1, S1, J1) are premises of the four rewrite rules of \Cref{subsec:plan_rewrites}, i.e., selection pushdown through Extend and through Semantic Join, and Extend--Selection and Extend--Join
fusion, seven justify the execution and costing mechanisms of Section~4 and the remaining eight are left for future work.

\begin{table}[t]
\caption{\small Properties of the NousDB operators.
Role: \roleR, premise of a rule of \Cref{subsec:plan_rewrites}.
\roleX, justifies an execution mechanism, \roleF, not exploited by the current optimizer.}
\label{tab:properties}
\vspace{-10pt}
\scriptsize
\setlength{\tabcolsep}{3pt}
\begin{tabularx}{\columnwidth}
{@{}l>{\raggedright\arraybackslash}X
  >{\raggedright\arraybackslash}p{3.2cm}@{}}
\toprule
\textbf{\#} & \textbf{Property} &
\textbf{Role in NousDB}\\
\midrule
E1 & $\sch(\ext(R))=\sch(R)\cup\attr E$;\; $\pi_{\sch(R)}(\ext(R))=R$
   & \roleR: premise of \ref{prop:sel-push-extend},
     \ref{prop:ext-sel-fusion}, \ref{prop:ext-join-fusion}\\
E2 & $|\ext(R)|=|R|$
   & totality: $\lceil |R|/k\rceil$ calls, $|R|$ rows\\
E3 & $R\subseteq S\Rightarrow \ext(R)\subseteq\ext(S)$
   & \roleX: pipelining; \texttt{LIMIT} early stop\\
E4 & $\ext(R\cup S)=\ext(R)\cup\ext(S)$
   & \roleX: (re)partitioning, batching\\
E5 & $\ext(R)=\ext_{\pi(n)}\circ\dots\circ\ext_{\pi(1)}(R)$ for any
     permutation $\pi$, where
     $\ext_i=\ext_{\attr E_i,\attr B,m,p_i}$,
     $\attr E=\biguplus_i\attr E_i$
   & \roleF: split or merge multi-attribute Extends\\
\midrule
S1 & $\sch(\sigma(R))=\sch(R)$;\; $\sigma(R)\subseteq R$
   & \roleR: premise of \ref{prop:sel-push-extend},
     \ref{prop:sel-push-join}\\
S2 & $R\subseteq S\Rightarrow\sigma(R)\subseteq\sigma(S)$
   & \roleX: early stop under \texttt{LIMIT}\\
S3 & $\sigma(R\cup S)=\sigma(R)\cup\sigma(S)$
   & \roleX: (re)partitioning, batching\\
S4 & $\sigma(R\cap S)=\sigma(R)\cap\sigma(S)$;\;
     $\sigma(R\setminus S)=\sigma(R)\setminus S=\sigma(R)\setminus\sigma(S)$
   & --\\
S5 & $\sigma(\sigma(R))=\sigma(R)$
   & --\\
S6 & $\sigma_1(\sigma_2(R))=\sigma_2(\sigma_1(R))
     ={\{r\in R\mid P_1(r)\wedge P_2(r)\}}$
   & \roleF: cost-based ordering of selections\\
\midrule
J1 & $\sch(R\semjoin S)=\sch(R)\cup\sch(S)$;\;
     $R\semjoin S\subseteq R\times S$
   & \roleR: premise of \ref{prop:sel-push-join},
     \ref{prop:ext-join-fusion}; \roleC:~${\le}\,|R|\,|S|$~pairs\\
J2 & $R_1\subseteq R_2,\ S_1\subseteq S_2\Rightarrow
     R_1\semjoin S_1\subseteq R_2\semjoin S_2$
   & \roleX: early stop under \texttt{LIMIT}\\
J3 & $(R_1\cup R_2)\semjoin S=(R_1\semjoin S)\cup(R_2\semjoin S)$, also
     in $S$; so $R\semjoin S=\bigcup_{i,j}R_i\semjoin S_j$
   & \roleX: block nested loop over $b_L\times b_R$ blocks;
     repartitioning\\
J4 & As J3, for $\cap$ and $\setminus$
   & --\\
J5 & $R\semjoin_{\attr B,\attr C,m,p}S=S\semjoin_{\attr C,\attr B,m,p}R$
   & \roleF: choice of the outer input\\
J6 & $(R\semjoin_1 S)\semjoin_2 T=R\semjoin_1(S\semjoin_2 T)$
   & \roleF: semantic-join reordering\\
J7 & $\pi_{\attr X}(R\semjoin S)
     =\pi_{\attr X\cap\sch(R)}(R)\semjoin\pi_{\attr X\cap\sch(S)}(S)$
   & \roleF: projection pushdown below~$\semjoin$\\
\bottomrule
\end{tabularx}
\vspace{-10pt}
\end{table}





 \label{sec:data_model}
\section{NousDB Optimizer} 

To bridge DataFusion's row-oriented execution model with the token-oriented cost and behavior of model inference, NousDB introduces \textit{four} complementary optimizations:
(\textit{a})  \textit{Token-aware batch sizing} introducing equi-depth token histograms (Section \ref{subsec:token-histograms}) (\textit{b}) \textit{Batch-fidelity certification} for preserving output fidelity in batched execution (Section \ref{subsec:batch-optimization}), (\textit{c}) \textit{token repartition operator} for balancing token load across partitions (Section \ref{sec:token-repartitioning}), and (\textit{d}) 
 \textit{algebraic rewrite rules} to cut model invocations through semantic selection pushdown and operator fusion (Section \ref{subsec:plan_rewrites}).

 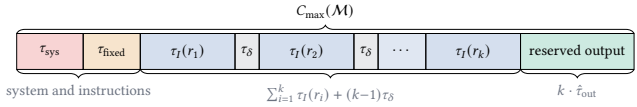
\begin{figure}[t]
\centering
\resizebox{\columnwidth}{!}{
    \begin{tikzpicture}[
        segment/.style={minimum height=0.75cm, font=\normalsize, align=center,
                        inner sep=2pt, anchor=west},
        brace/.style={decorate, decoration={brace, amplitude=5pt, mirror}},
        label/.style={font=\normalsize, text=nousgray}
    ]
        \def\barheight{0.3cm}
        \def\totalwidth{14cm}
    
        \node[segment, fill=nousred!20, minimum width=1.4cm] (sys) at (0,0)
            {$\tau_{\mathrm{sys}}$};
        \node[segment, fill=nousorange!20, minimum width=1.2cm, right=0pt of sys]
            (fixed) {$\tau_{\mathrm{fixed}}$};
        \node[segment, fill=nousblue!20, minimum width=2.0cm, right=0pt of fixed]
            (r1) {$\tau_I(r_1)$};
        \node[segment, fill=nousgray!15, minimum width=0.5cm, right=0pt of r1]
            (d1) {$\tau_\delta$};
        \node[segment, fill=nousblue!20, minimum width=2.0cm, right=0pt of d1]
            (r2) {$\tau_I(r_2)$};
        \node[segment, fill=nousgray!15, minimum width=0.5cm, right=0pt of r2]
            (d2) {$\tau_\delta$};
        \node[segment, fill=nousblue!15, minimum width=1.0cm, right=0pt of d2]
            (dots) {$\cdots$};
        \node[segment, fill=nousblue!20, minimum width=2.0cm, right=0pt of dots]
            (rk) {$\tau_I(r_k)$};
        \node[segment, fill=nousgreen!20, minimum width=2.4cm, right=0pt of rk]
            (out) {reserved output};
    
        \foreach \n in {sys,fixed,r1,d1,r2,d2,dots,rk,out} {
            \draw (\n.south west) rectangle (\n.north east);
        }
    
        \draw[brace] (sys.south west) -- node[below=6pt, label] {system and instructions}
            (fixed.south east);
        \draw[brace] (r1.south west) -- node[below=6pt, label]
            {$\sum_{i=1}^k \tau_I(r_i) + (k{-}1)\tau_\delta$} (rk.south east);
        \draw[brace] (out.south west) -- node[below=6pt, label]
            {$k \cdot \hat{\tau}_{\mathrm{out}}$} (out.south east);
    
        \draw[decorate, decoration={brace, amplitude=5pt}]
            (sys.north west) -- node[above=6pt]
            {$\Cmax(\mathcal{M})$} (out.north east);
    \end{tikzpicture}
    }
\vspace{-10pt}
\caption{\small LLM context window decomposition.
     System $\tau_{\mathrm{sys}}$ and user prompt $\tau_{\mathrm{fixed}}$ are followed by per-row input tokens $\tau_I(r_i)$ }
\label{fig:token-budget}
\vspace{-10pt}
\end{figure} \label{fig:context_window}

\textbf{LLM Context Window}. \label{sec:context-window} 
Every model $m$ imposes a maximum context window $C_{\max}(m) \in \mathbb{N}$,  bounding the total input plus output tokens that a single invocation can process.  
The input prompt comprises fixed overhead, a system prompt  ($\tau_{\mathit{sys}}$ tokens) and an instruction template ($\tau_{\mathit{inst}}(p)$ tokens), followed by the serialized input tuples, each separated by a delimiter  of $\tau_{\delta}$ tokens (Figure~\ref{fig:token-budget}).  
We denote by $\hat{\tau}_{\mathit{out}}$ the expected output tokens per processed tuple.

\begin{definition}[Batch Input Token Decomposition]
\label{def:batch-input}
The total input tokens for a batch $B$ with $|B| = k$ tuples decomposes as:
\begin{equation}
    \tau_{\text{in}}(B) = \underbrace{\tau_{\text{sys}}}_{\text{system prompt}} + \underbrace{\tau_{\text{inst}}(p)}_{\text{instruction}} + \underbrace{\sum_{r \in B} \tau_I(r)}_{\text{data tokens}} + \underbrace{(k-1) \cdot \tau_{\text{delim}}}_{\text{delimiters}}
\end{equation}
\end{definition}

%
%

Every batch must satisfy the \emph{token budget constraint}:
\begin{equation}
\label{eq:budget-constraint}
\tau_{\mathit{in}}(B) + k \cdot \hat{\tau}_{\mathit{out}} \;\leq\; C_{\max}(m)
\end{equation}
Since total output tokens cannot be known a priori, we reserve a fixed fraction 
of the context window for output predictions, in our case the $30\%$ of the context window.  
This constraint is the foundation of the batch size optimization: maximizing $k$ amortizes the fixed overhead 
$\tau_{\mathit{sys}}+\tau_{\mathit{inst}}(p)$ across more tuples, reducing both 
the number of LLM invocations and token consumption.


\subsection{Token Distributions and Token Histograms}
\label{subsec:token-histograms}



Since tokens are the fundamental unit of computation for LLMs, to capture the \emph{token count distribution}, we introduce \emph{token histograms}, adopting the abstract partitioning framework of Ioannidis et al. ~\cite{10.1145/235968.233342}.


Let $\mathcal{T}_m$ denote the tokenizer of model $m$, and let $\tau_m(s)=|\mathcal{T}_m(s)|$ be the number of tokens in string $s$ using $\mathcal{T}_m$. When $m$ is clear, we write $\tau$. 
For a tuple $r$ in relation $R$ with schema $\mathcal{A}$ and input attribute set $I \subseteq \mathcal{A}$ used for model inference, the \emph{tuple token count} is $ \tau_I(r) = \tau(r[I])$. 
For attribute $A$ in $R$, let $V_A = \{v_1, \ldots, v_D\}$ be the set of $D$ distinct values in $A$, with corresponding frequencies $f_1, \ldots, f_D$ where $f_i = |\{r \in R : r[A] = v_i\}|$. The \emph{token distribution} of $A$ is 
   $\mathcal{T}_A = \{(\tau(v_1), f_1), (\tau(v_2), f_2), \ldots, (\tau(v_D), f_D)\}$.

 \begin{definition}[Cumulative Token Distribution]
\label{def:cumulative-token-dist}
The \emph{cumulative token distribution} function $C_\tau: \mathbb{N} \rightarrow [0,1]$ gives the fraction of tuples with token count at most $t$:
\begin{equation}
    C_\tau(t) = \frac{1}{N} \sum_{i: \tau(v_i) \leq t} f_i
\end{equation}
where $N = \sum_{i=1}^{D} f_i = |R|$ is the total number of tuples.
\end{definition}

\label{subsubsec:foundations}

\begin{table}[t]
\scriptsize
\caption{\small Notation for batching and token repartitioning}
\vspace{-10pt}
\setlength{\tabcolsep}{3pt}
\renewcommand{\arraystretch}{1.1}
\begin{tabularx}{\columnwidth}
{@{}l
  >{\raggedright\arraybackslash}X@{}}
\toprule
\textbf{Symbol} &
\textbf{Meaning}\\
\midrule
$\tau_I(r)$ & token count of tuple $r$ over the input attributes $I$\\
$C_{\max}(m)$ & \emph{context window} of model $m$, input plus output tokens of one invocation\\
$\tau_{\text{sys}}$, $\tau_{\text{inst}}$ & tokens of the system prompt and the instruction template, paid once per invocation\\
$\tau_{\text{delim}}$, $\hat{\tau}_{\text{out}}$ & delimiter tokens and expected output tokens, per tuple\\
$H_\tau$, $\beta$ & equi-depth \emph{token histogram} of the input and its number of buckets (\Cref{subsec:histogram-def})\\
$\mu_\tau$, $Q_p$ & mean and $p$-quantile of the token count, read from $H_\tau$\\
$\alpha$ & \emph{overflow probability}, the chance that a batch exceeds $C_{\max}(m)$\\
$k_{H_\tau}$ & window-safe batch size (\Cref{alg:batch-size}), the largest candidate\\
\midrule
$y^{*}(r)$ & \emph{per-row output}, the output of tuple $r$ invoked alone\\
$y_B(r)$ & output of $r$ inside a batch $B$ of $k$ tuples (pairs, for a join \\
$\varphi(k)$ & \emph{batch fidelity} at size $k$, $\Pr[\,y_B(r) = y^{*}(r)\,]$\\
$\eta$, $\delta$ & \emph{tolerance} (fraction of tuples that may deviate from $y^{*}$) and \emph{confidence} ($1-\delta$)\\
\midrule
$\mathcal{K}$, $L$ & \emph{candidate sizes} $\{2^j : j \ge 1,\ 2^j < k_{H_\tau}\} \cup \{k_{H_\tau}\}$ and their number $|\mathcal{K}|$\\
$r_1,\dots,r_n$ & \emph{probe tuples}, a uniform sample of the input; $n$ is the \emph{probe budget}, Eq.~\eqref{eq:probe-budget}\\
$B_j$ & \emph{probe batch}, a batch that contains the probe tuple $r_j$ in a random slot\\
$X_k$ & \emph{agreement count}, the number of probes with $y_{B_j}(r_j) = y^{*}(r_j)$\\
\midrule
$\delta'$ & error budget of one candidate, $\delta/L$\\
$\underline{\varphi}_{\delta'}(X_k, n)$ & Clopper--Pearson lower confidence bound on $\varphi(k)$, Eq.~\eqref{eq:cp-bound}\\
$k_{safe}$ & \emph{fidelity-certified batch size}, the largest certified candidate\\
$k^*$ & executed batch size, $\min(k_{H_\tau}, k_{safe})$\\
\midrule
$P$, $R_j$ & number of partitions and the $j$-th partition of relation $R$\\
$\rho^{\mathcal{S}}_\tau$ & \emph{token repartition operator} with strategy $\mathcal{S}$, greedy or zigzag (\Cref{def:token-repartition})\\
$T(R_j)$ & token load of partition $R_j$, $\sum_{r \in R_j} \tau_I(r)$\\
$\lambda$ & \emph{imbalance ratio}, $\max_j T(R_j) / \min_j T(R_j)$\\
$\mathrm{CV}_\tau$ & token skew, the coefficient of variation $\sigma_\tau/\mu_\tau$ of the per-tuple token counts\\
$T_{\mathit{avg}}$, $\taumax$ & mean partition load and largest token count of a tuple\\
\bottomrule
\end{tabularx}

\label{tab:notation}\label{tab:fidelity-notation}
\vspace{-10pt}
\end{table}





\begin{definition}[Token Histogram]
\label{def:token-histogram}
A token histogram $H_\tau$ on attribute $A$ has $\beta \geq 1$ buckets and is constructed by sorting the token distribution $\mathcal{T}_A$ according to token counts and frequencies, partitioning it into $\beta$ disjoint buckets $\{B_1,\ldots,B_\beta\}$ subject to a partitioning constraint, and storing summary statistics for each bucket. \\
Each bucket $B_j$ stores: (\textit{a}) the minimum and maximum token counts, $\tau_{\min}(B_j)$ and $\tau_{\max}(B_j)$, resp. (\textit{b}) the number $n_j$ of distinct attribute values (\textit{c}) the total frequency $F_j = \sum_{v_i \in B_j} f_i$, i.e., the number of tuples represented, and (\textit{d}) the frequency-weighted mean token count $\bar{\tau}_j = \frac{1}{F_j}\sum_{v_i \in B_j} f_i \tau_A(v_i)$.
\end{definition}

\noindent The choice of partition constraint yields different histogram families. 
For NousDB we adopt equi-depth partitioning, where each bucket is constrained to contain approximately the same number of tuples.


\begin{definition}[Equi-Depth Token Histogram]
\label{def:equi-depth-histogram}
An equi-depth token histogram with $\beta$ buckets partitions the token distribution $\mathcal{T}_A$ such that each bucket contains approximately $N/\beta$ tuples, i.e.:  $F_j \approx N/\beta, j \in \{1, \ldots, \beta\}$.
The bucket boundaries are defined by the token count quantiles. Let $q_0 = 0$ and, for $j \in \{1, \ldots, \beta\}$:
\begin{equation}
q_j = \min\!\Bigl\{\tau \geq 0 \;:\; C_\tau(\tau) \geq \frac{j}{\beta}\Bigr\}\,,
\label{eq:bucket-boundary}
\end{equation}
where $C_\tau$ is the cumulative token distribution from \Cref{def:cumulative-token-dist}. Bucket $B_j$ contains all values $v$ with $q_{j-1} < \tau_A(v) \leq q_j$.
\end{definition}

We chose equi-depth histograms for three reasons: 

\noindent
\textbf{a) Quantile access.}
Batch size optimization (\Cref{subsec:batch-optimization}) requires quantile queries of the form ``what token count is exceeded by at most $\alpha\%$ of tuples?'' In an equi-depth histogram, the boundary $q_j$ between buckets $j$ and $j+1$ is precisely the $(j/\beta)$-quantile, so these queries reduce to bucket lookups with intra-bucket interpolation.

\noindent
\textbf{(b) Tail estimation.}
Exceeding the context limit $C_{\max}(m)$ causes runtime failure. Accurate tail-quantile estimates are therefore critical for preventing overflow. Equi-depth histograms support such estimates by isolating high-token-count values in dedicated bins.

\noindent
\textbf{(c) Sum estimation.}
For a batch of $k$ tuples, the cost model estimates $\sum_{i=1}^{k} \tau_I(r_i)$. Equi-depth histograms enable efficient sum estimation via quantile-based bounds, since each quantile range contributes proportionally to the total token mass.

Given an equi-depth histogram, we extract arbitrary quantiles via linear interpolation within the target bucket.

\begin{definition}[Quantile Extraction from Equi-Depth Histograms]
\label{def:quantile-extraction}
For an equi-depth histogram $H_\tau$ with $\beta$ buckets, the $p$-quantile $Q_p$ (for $p \in (0,1]$) is estimated as:
\begin{equation}
Q_p \;\approx\; \tau_{\min}(B_j) \;+\; \bigl(p \cdot \beta - (j-1)\bigr) \cdot \bigl(\tau_{\max}(B_j) - \tau_{\min}(B_j)\bigr)
\label{eq:quantile-extraction}
\end{equation}
where $j = \lceil p \cdot \beta \rceil$ is the index of the bucket containing the $p$-quantile. The expression $p \cdot \beta - (j-1) \in [0,1]$ is the fractional position within bucket~$B_j$, and the formula assumes a uniform distribution of token counts within each bucket.
\end{definition}

\textbf{Histogram Summary Statistics.}
\label{sec:hist-summary-stats} 
From the per-bucket statistics stored in $H_\tau$, we derive two global summary statistics that are consumed directly by the batch sizing algorithm (\Cref{subsec:batch-optimization}) and the token repartitioning operator (\Cref{sec:token-repartitioning}), namely  the \emph{mean token count} $\mu_\tau$ and \emph{total tokens} $T_{\text{total}}$:
\begin{equation}
\mu_\tau = \frac{1}{N} \sum_{j=1}^{\beta} F_j \cdot \bar{\tau}_j, \qquad T_{\text{total}} = N \cdot \mu_\tau
\label{eq:histogram-stats}
\end{equation}





The mean $\mu_\tau$ provides the initial batch size estimate (\Cref{{alg:batch-size}}), while $T_{\text{total}}$ determines the per-partition token budget for the repartitioning operator (\Cref{def:token-repartition}).

\textbf{Construction}. Histograms are built once per table and text attribute, offline, by one tokenization pass over the attribute, and stored in a catalog that the optimizer loads at planning time.
When the input attribute set $I$ of an operator has several attributes, $\tau_I(r)$ is the token count of the whole rendered row, and the quantile of a sum is not the sum of the per-attribute quantiles.
NousDB therefore rebuilds one histogram for $I$ by summing the per-attribute bucket profiles per bin, i.e., under a comonotone coupling of the attributes, for which quantiles are additive.

\subsection{Batch Size Optimization}
\label{subsec:batch-optimization}

\begin{algorithm}[t]
    \footnotesize                                  
    \setlength{\abovecaptionskip}{2pt}      
    \setlength{\belowcaptionskip}{0pt}
    \caption{\small Iterative Safe Batch Size Computation}
    \label{alg:batch-size}
    \begin{algorithmic}[1]
    \setstretch{0.95}                      
    \Require $H_\tau$; $C_{\max}$; $\tau_{\text{sys}}+\tau_{\text{inst}}$;
             $\hat{\tau}_{\text{out}}$; confidence $1-\alpha$
    \Ensure {Safe batch size $k_{H_\tau}$}
    \State $C_{\text{avail}} \gets C_{\max} - \tau_{\text{sys}} - \tau_{\text{inst}}$
    \State $\mu_\tau \gets \textsc{Mean}(H_\tau)$
    \State $k \gets \lfloor C_{\text{avail}} / (\mu_\tau + \hat{\tau}_{\text{out}}) \rfloor$ \Comment{mean-based init}
    \For{$i = 1$ \textbf{to} $\textsc{MaxIter}$}
        \State $p \gets 1 - \alpha/k$
        \State $q \gets \textsc{GetQuantile}(H_\tau, p)$
        \State $k_{\text{new}} \gets \lfloor C_{\text{avail}} / (q + \hat{\tau}_{\text{out}}) \rfloor$
        \If{$k_{\text{new}} = k$} \State \Return $k$ \EndIf
        \State $k \gets k_{\text{new}}$
    \EndFor
    \State \Return $k$
    \end{algorithmic}
\end{algorithm}


The operator definitions in \Cref{sem_ops} assume singleton execution, under which each tuple's output depends only on that tuple (\emph{input locality}). Batching reduces the $|R|$ singleton invocations to $\lceil |R|/k \rceil$ invocations of up to $k$ tuples, amortizing per-call overhead to reduce token usage and improve throughput~\cite{lin2024batchprompt, Cheng2023BatchPE, 10.1145/3696410.3714950, 10.14778/3734839.3734853}. However, a batch must both fit within the context window and preserve singleton outputs. Batching is thus a physical optimization valid only when it preserves the logical operator semantics.




To define safe batching, we distinguish the two outputs of a tuple under a fixed operator, model mmm, and prompt template.

\begin{definition}[Batched Invocation and Batch-Composition Invariance]\label{def:batched-invocation}\label{def:invariance}
    Let $B = \langle r_1,\dots,r_k \rangle$ be a batch of $k$ tuples.
    The prompt constructor renders one prompt $pc(B)$ that holds all $k$ tuples, and one invocation $m(pc(B))$ returns one output per tuple.
    We write $y_B(r)$ for the output returned for $r \in B$, and
    \begin{equation}
            y^{*}(r) \;=\; m\bigl(pc(r)\bigr)
        \label{eq:reference-output}
    \end{equation}
    for its \emph{per-row output}, the value that Definitions~\ref{extend_op}--\ref{join_op} assign to $r$.
    The operator is \emph{batch-composition invariant} at size $k$ if $y_B(r) = y^{*}(r)$ for every batch $B$ of $k$ tuples and every $r \in B$.
\end{definition}


The batch size $k$ is therefore the central physical decision for a semantic operator, and its choice is governed by two constraints: 

\noindent\textit{(i) Safe context-window utilization.}
The $k$ tuples and their outputs must fit in the context window $C_{\max}(m)$, since a batch that does not fit  is not rejected or truncated, but results to query failure.
Because a batch fit depends only on the token lengths of its tuples, the optimizer can   determine the maximum feasible batch size at planning time using token histograms, \textit{without invoking the model}.

\noindent\textit{(ii) Safe batching against per-row invocation.}
A batch can affect per row results: the output a tuple receives inside a batch may differ from the one it receives when invoked alone.
This failure is \emph{silent}, and it depends on the model, the prompt and the data, so it can only be measured\textit{ by invoking the model}.

The two problems are orthogonal. The first concerns the relationship between tuple token lengths and the window capacity, while the second concerns the model behavior. Neither implies the other: batches that fill a small fraction of the window may already lose quality.
\Cref{fig:distributions_and_sweep} illustrates this distinction for two cases: IMDB and SemBench reviews. The top row shows the token-length distributions of the review columns, which vary substantially across datasets. For example, the 90th-percentile review length in SemBench is approximately equal to the mean review length in IMDB.
In the bottom row, we execute four queries from IMDB under different batch size $k$ and measure their quality:
the candidate batch sizes are $k \in \{1,2,4,8,16,32,64\}$, starting from singleton execution and doubling at each step, plus two larger settings that utilize the $50\%$ and $75\%$ of the model's context window.
The results are mixed: in some cases batching provides a boost in the total accuracy, like Q1, Q4 , while it can be catastrophic,  like in Q2 and Q7.

\begin{definition}[Maximum Feasible Batch Size]
\label{def:max-batch}
The \emph{maximum feasible batch size} $k_{H_\tau}$, given a token histogram $H_\tau$, is the largest $k$ such that the expected batch, based on its estimated token demand, fits within the context window:
$\tau_{\text{sys}} + \tau_{\text{inst}} + k \cdot (\mu_\tau + \tau_{\text{delim}}) + k \cdot \hat{\tau}_{\text{out}} \leq C_{\max}(m)$. 
Solving for $k$:
\begin{equation} \label{eq:kmax}
    k_{H_\tau} = \left\lfloor \frac{C_{\max}(m) - \tau_{\text{sys}} - \tau_{\text{inst}}}{\mu_\tau + \tau_{\text{delim}} + \hat{\tau}_{\text{out}}} \right\rfloor
\end{equation}
\end{definition}


\noindent Recall that $\mu_{\tau}$ is the estimated mean number of input tokens per tuple from $H_\tau$, $\tau_{\mathrm{delim}}$ is the delimeter tokens, and $\hat{\tau}_{\mathrm{out}}$ is the estimated number of output tokens per tuple.
\Cref{alg:batch-size} computes $k_{H_\tau}$ in three steps, without invoking the model.
It first subtracts the fixed prompt from the window, which leaves the budget $C_{\text{avail}}$ for the tuples and their outputs, and starts from the mean-based size of Equation~\eqref{eq:kmax} (lines 1--3).
This start is optimistic: it sizes the batch for the \emph{mean} tuple, so any batch of longer-than-average tuples overflows.
Each iteration therefore sizes the batch for a long tuple instead, the $(1-\alpha/k)$-quantile $q$ of $H_\tau$ (lines 5--7): a batch of $k$ tuples of at most $q$ tokens each fits, and by the union bound one of its $k$ tuples exceeds $q$ with probability at most $k \cdot \alpha/k = \alpha$, whether or not the tuple lengths are independent.
Finally, $q$ depends on $k$: larger batches are more likely to contain long tuples and therefore require higher quantiles, which in turn reduces the feasible batch size.
The algorithm therefore repeats until $k$ stops changing (lines 8--11) and returns this fixed point as $k_{H_\tau}$, a batch size that overflows the window with probability at most $\alpha$.



\begin{figure}[t]
  \centering
  \includegraphics[width=\columnwidth]{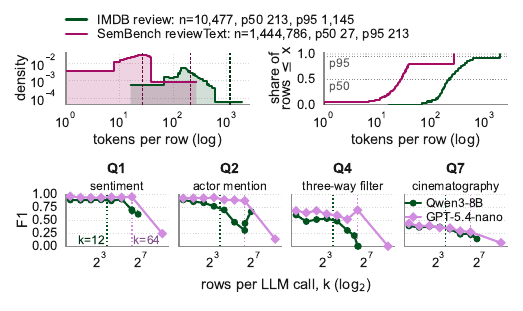}
  \vspace{-25pt}
  \caption{\small Reviews tokens distribution (up) and impact of batching in quality (bottom)}
  \vspace{-15pt}
  \label{fig:distributions_and_sweep}
\end{figure}

\subsubsection{Fidelity-Aware Batch Size} \label{batch_fidelity}
\Cref{alg:batch-size} bounds batch capacity, but does not ensure that the model preserves its singleton behavior in the batch. 
The bottom row of \Cref{fig:distributions_and_sweep} shows that sensitivity to\textit{ batching varies across queries and models}. With Qwen3-8B, the F1 score of Q2 falls from 0.89 at $k=1$ to 0.77 at $k=8$ and 0.30 at $k=64$. In contrast, for Q1, the F1 score stays between 0.87 and 0.90 up to $k=32$, dropping to 0.70 only at $k=64$. With GPT-5.4-nano, the F1 scores of both queries remain within 0.05 of their singleton values up to $k=64$, but collapse at $75\%$ of GPT's context window reaching 0.24 for Q1 and 0.14 for Q2. 
The top row of \Cref{fig:distributions_and_sweep} reveals additionally the role of \textit{tuple length}. an IMDB review spans 213 tokens at the median and 1,145 at the 95th percentile, against 27 and 213 for a SemBench review, so the same $k$ packs about eight times more text into one invocation on IMDB, and batches of equal $k$ differ in length even within one column.

Estimating the largest batch size such that output quality is close to that of singleton execution, at query time has three challenges.
\textit{(a)} Because no ground truth is available, we measure agreement against the engine’s singleton outputs, $y^{*}(r)$.
This follows prior work, which states statistical guarantees for LLM-powered operators relative to a gold algorithm~\cite{lotus,10.1145/3769776} and evaluates batched plans against unbatched execution~\cite{qi_vldb26,Liu2024CliqueParcelAA}.
 \textit{(b)} The $k$ outputs produced by a batch share the same context and may therefore be dependent, hence they cannot be treated as independent samples.
 \textit{(c)} The solution must be substantially cheaper than the query it protects, and its cost must not grow with $|R|$.
We therefore ask: given a semantic operator, an error tolerance $\eta\in(0,1)$, and a failure probability $\delta\in(0,1)$, what is the largest batch size $k$ such that, with probability at least $1-\delta$, the batched outputs differ from the corresponding per-row outputs on at most an $\eta$ fraction of tuples?


\begin{definition}[Batch Fidelity]
\label{def:fidelity}
    Given a batch $B$ that ranges over the batches of $k$ tuples that the engine forms,  the \emph{batch fidelity} of an operator at size $k$ is the probability that a tuple receives its per-row output inside a batch of $k$ tuples:
    \begin{equation}
        \varphi(k) \;=\;
        \Pr_{\,B,\; r \sim \mathrm{Unif}(B)}
        \bigl[\, y_B(r) = y^{*}(r) \,\bigr],
        \qquad \varphi(1) = 1,
        \label{eq:fidelity}
    \end{equation}
\end{definition}

To measure fidelity, a natural approach is to sample tuples and compare their batched outputs with their singleton outputs. However, outcomes for tuples within the same batch are dependent because the tuples share a context and may influence one another. We therefore evaluate one probe tuple per independently constructed batch. For a fixed batch size $k$, each probe defines a Bernoulli trial that succeeds when its batched output agrees with its singleton output. Assuming that probe batches are sampled independently, the agreement count follows \[ X_k \sim \mathrm{Binomial}\!\left(n,\varphi_o(k)\right), \] where $\varphi_o(k)$ is the fidelity of operator $o$ at batch size $k$, $n$ is the probe budget, and $X_k \in \{0,\ldots,n\}$ is the number of successful probes.

More specifically, NousDB estimates fidelity through a calibration run over a sample of the operator input. In addition to $y^{*}$, $y_B$ (\Cref{def:invariance}), $\eta$, and $\delta$, calibration uses: \squishlist 
\item the set of \emph{candidate batch sizes} \[ \mathcal{K} = \{2^q : q \geq 1,\ 2^q < k_{H_\tau}\} \cup \{k_{H_\tau}\}, \qquad L = |\mathcal{K}|, \] which doubles from $2$ up to the histogram-derived maximum feasible batch size $k_{H_\tau}$;

\item the \emph{probe tuples} $r_1,\ldots,r_n$, sampled uniformly from the operator input, where $n$ is the \emph{probe budget} fixed before calibration; 

\item for each probe $r_i$ and candidate size $k \in \mathcal{K}$, an independently constructed \emph{probe batch} $B_{i,k}$, which places $r_i$ in a uniformly random position and fills the remaining $k-1$ positions with tuples sampled from the operator input; and 

\item the \emph{agreement count} $X_k = \sum_{i=1}^{n} \mathbb{I}\!\left[ y_{B_{i,k}}(r_i)=y^{*}(r_i) \right].$
\squishend

\noindent\textbf{Certifying a batch size.}
The fraction $X_k/n$ estimates $\varphi(k)$ without a guarantee, so certification uses a one-sided \emph{lower confidence bound}: overestimating fidelity certifies a size that silently corrupts outputs, whereas underestimating it merely forfeits speed.
The exact Clopper--Pearson bound \cite{10.1093/biomet/26.4.404} is the smallest fidelity $\delta'$ under which $X_k$ or more agreements remain $\delta'$-plausible,
\begin{equation}
    \underline{\varphi}_{\delta'}(X_k, n)
    = \mathrm{Beta}^{-1}\bigl(\delta';\; X_k,\; n - X_k + 1\bigr),
    \label{eq:cp-bound}
\end{equation}
and satisfies $\Pr[\underline{\varphi}_{\delta'} \le \varphi(k)] \ge 1-\delta'$ for every true fidelity.
A candidate $k$ is certified if $\underline{\varphi}_{\delta'}(X_k, n) \ge 1-\eta$, and spending $\delta' = \delta/L$ on each of the $L$ candidates bounds the probability that any certificate is wrong by $\delta$ (union bound).

\begin{algorithm}[t]
    \footnotesize
    \setlength{\abovecaptionskip}{2pt}
    \setlength{\belowcaptionskip}{0pt}
    \caption{\small Batch-Fidelity Certification}
    \label{alg:fidelity}
    \begin{algorithmic}[1]
    \setstretch{0.95}
    \Require input tuples; $k_{H_\tau}$; tolerance $\eta$; confidence $1-\delta$
    \Ensure Fidelity-certified batch size $k_{safe}$
    \State $\mathcal{K} \gets \{2^j : j \ge 1,\ 2^j < k_{H_\tau}\} \cup \{k_{H_\tau}\}$; $L \gets |\mathcal{K}|$; $\delta' \gets \delta/L$
    \State $n \gets \bigl\lceil \ln(1/\delta') \,/\, \ln\bigl(1/(1-\eta)\bigr) \bigr\rceil$ \Comment{Eq.~\eqref{eq:probe-budget}}
    \State draw probes $r_1,\dots,r_n$ and fillers from the input
    \State $y^{*}_j \gets m(pc(r_j))$ for $j = 1,\dots,n$ \Comment{$n$ singleton calls}
    \For{$k \in \mathcal{K}$ in descending order}
        \State $B_j \gets$ $r_j$ in a random slot of $k$, rest fillers, $j = 1,\dots,n$
        \State $X_k \gets |\{\, j : y_{B_j}(r_j) = y^{*}_j \,\}|$ \Comment{$n$ batched calls}
        \If{$\mathrm{Beta}^{-1}(\delta';\, X_k,\, n-X_k+1) \ge 1-\eta$} \State \Return $k$ \EndIf
    \EndFor
    \State \Return $1$
    \end{algorithmic}
\end{algorithm}

We call $k_{safe}$, the batch size that preserves the fidelity of the batch and \Cref{alg:fidelity} computes it  once per semantic operator inside a query.
It obtains the $n$ singleton outputs once (line 4) and tries the candidates from the largest down (lines 5--8), starting from the batch size $k_{H_\tau}$, provided by \Cref{alg:batch-size}.
A candidate is abandoned as soon as its disagreements make a pass impossible, which saves model calls.
Finally, we must determine the probe budget $n$, i.e., the number of tuples needed to test to measure engine's fidelity.
If all probes agree ($X_k = n$), Equation~\eqref{eq:cp-bound} attains $\underline{\varphi}_{\delta'} = (\delta')^{1/n}$, so a perfect count passes once
\begin{equation}
n \;\ge\; \frac{\ln(L/\delta)}{\ln\bigl(1/(1-\eta)\bigr)}
\;\approx\; \frac{\ln(L/\delta)}{\eta};
\label{eq:probe-budget}
\end{equation}
at $(\eta,\delta) = (0.05, 0.05)$ and $L = 5$ this gives $n = 90$.
Calibration then costs at most $n$ singleton and $nL$ batched calls; the budget depends only on $(\eta, \delta, L)$ and not on $|R|$, so, like histogram construction, its share of the query cost shrinks as the data grows.

\noindent\textbf{Final batch size.}
Since a candidate batch size is at most $k_{H_\tau}$, the executed size $k^* = \min(k_{H_\tau}, k_{safe})$ is $k_{safe}$ whenever certification runs ($k_{H_\tau}$ otherwise).
The plan thus carries one guarantee per problem of \Cref{subsec:batch-optimization}: each batch overflows the window with probability at most $\alpha$, and fidelity is at least $1-\eta$ with confidence $1-\delta$.

\textbf{Example.}
Consider the sentiment selection of \Cref{fig:distributions_and_sweep} over IMDB reviews, with $pt$ rendering \emph{``Does \{review\} express positive sentiment? Answer yes or no.''}.
The default values for the failure probability $\delta$ and fraction of tuples $\eta$ are $\eta = \delta = 0.05$, meaning at most a $5\%$ tuple-level deviation from reference semantics, wrong with probability at most $5\%$. 
The first step of the optimizer, predicts the histogram guided batch size from \Cref{alg:batch-size}, which yields $k_{H_\tau} = 32$, hence the set if candidate batch sizes are $\mathcal{K} = \{2, 4, 8, 16, 32\}$, $L = 5$, $\delta' = \delta / L = 0.01$, and the probe budget $n = 90$ by Equation~\eqref{eq:probe-budget}. 
The descending scan first tests $k = 32$: calibration returns $X_{32} = 81$, an empirical fidelity of $0.90$, and Equation~\eqref{eq:cp-bound} gives $\underline{\varphi}_{0.01} = \mathrm{Beta}^{-1}(0.01; 81, 10) = 0.803 < 0.95$, so the candidate fails. 
Next, at $k = 16$ every probe agrees, $X_{16} = 90$, the bound attains its closed form $\underline{\varphi}_{0.01} = 0.01^{1/90} = 0.950 \ge 0.95$, and the scan stops. 
Thus $k_{safe} = 16$ and $k^{*} = \min(32, 16) = 16$: 
with confidence at least $0.95$, the fidelity at $k=16$ is at least $0.95$. Calibration requires $90$ singleton calls and $180$ batched calls, corresponding to the two candidates tested. This cost is independent of $|R|$.


\subsection{Token-Based Repartitioning} \label{sec:token-repartitioning}


Modern query engines distribute rows uniformly across parallel workers. In semantic query engines, however, cost is driven by token volume, so row-balanced partitions may be token-skewed. We quantify this risk using the coefficient of variation of per-tuple token counts, $\mathrm{CV}_\tau=\sigma_\tau/\bar{\tau}$, where $\bar{\tau}$ and $\sigma_\tau$ are their mean and standard deviation. Larger values indicate greater variation in token load despite equal row counts.

\begin{definition}[Token repartition operator]\label{def:token-repartition}
    Given a relation $R$, the input attributes $I$ of the downstream semantic operator, $P$ output partitions and a strategy $\mathcal{S}$, the operator $\rho^{\mathcal{S}}_\tau(R,I,P)$ produces a partition $\{R_0,\dots,R_{P-1}\}$ of $R$ ($\bigcup_jR_j=R$, $R_i\cap R_j=\emptyset$) that assigns each tuple by its token count under $\mathcal{S}$, aiming to minimize the \emph{imbalance ratio} $\lambda=\max_jT(R_j)/\min_jT(R_j)$, where $T(R_j)=\sum_{r\in R_j}\tau_I(r)$.
\end{definition}

\textbf{Strategies} (\Cref{fig:partitioning-strategies}).
\emph{Greedy Bin-Packing} (the default) assigns each tuple to the partition with the fewest accumulated tokens, analogous to the Longest Processing Time (LPT) scheduling heuristic for identical parallel machines \cite{doi:10.1137/0117039}.
\emph{Zigzag} sorts the tuples of each wave by decreasing token count and deals them to the partitions in a boustrophedon order ($P_0,\dots,P_{P-1},P_{P-1},\dots,P_0,\dots$), for downstream operators that want equal row counts.

\begin{theorem}[Balance of Greedy Bin-Packing]\label{thm:greedy-bound}
With $T_{\mathit{avg}}=\frac1P\sum_jT(R_j)$ and $\taumax$ the largest token count, greedy bn-packing guarantees $\max_jT(R_j)-\min_jT(R_j)\le\taumax$ at every point of the stream and $\max_jT(R_j)\le T_{\mathit{avg}}+\taumax$; hence $\lambda\le1+\taumax/(T_{\mathit{avg}}-\taumax)$ whenever $T_{\mathit{avg}}>\taumax$.
Zigzag guarantees that the row counts of the partitions differ by at most one within each wave.
\end{theorem}
\begin{proof}
A tuple of size $\tau\le\taumax$ joins a partition of minimum load $L_{\min}$; the new maximum is either the old one or $L_{\min}+\tau$, and the new minimum is at least $L_{\min}$, so the spread stays at most $\taumax$ by induction.
The partition that attains the maximum received its last tuple while it was the least loaded, i.e.\ at most the average, which gives the second bound and the ratio follows from $\min_jT(R_j)\ge\max_jT(R_j)-\taumax\ge T_{\mathit{avg}}-\taumax$.
Zigzag deals the $n$ tuples of a wave cyclically, so each partition receives $\lfloor n/P\rfloor$ or $\lceil n/P\rceil$ of them.
\end{proof}

\begin{figure}[t]
\centering
\resizebox{\columnwidth}{130pt}{
\begin{tikzpicture}[
    cell/.style={draw, minimum height=0.33cm, minimum width=0.55cm, font=\tiny, inner sep=1pt},
    plabel/.style={font=\scriptsize\bfseries},
    arrow/.style={->, semithick},
]

\begin{scope}
    \node[font=\scriptsize\bfseries, anchor=west] at (-0.7, 1.15) {(a) Greedy Bin-Packing};

    \foreach \i/\tok in {0/100, 1/300, 2/50, 3/200, 4/150, 5/40} {
      \node[cell, fill=nousblue!15] (r\i) at (\i*1.1 + 0.7, 0.6) {$\tau{=}\tok$};
      \node[font=\tiny, above=0pt] at (r\i.north) {$r_{\i}$};
    }

    \node[plabel] at (-0.3, -0.15) {$P_0$};
    \node[plabel] at (-0.3, -0.5) {$P_1$};
    \node[plabel] at (-0.3, -0.9) {$P_2$};

    \draw[arrow, nousblue!70] (r0.south) -- ++(0,-0.35);
    \node[cell, fill=nousblue!20] at (0.7, -0.1) {100};

    \draw[arrow, nousblue!70] (r1.south) -- ++(0,-0.65);
    \node[cell, fill=nousblue!20] at (1.8, -0.45) {300};

    \draw[arrow, nousblue!70] (r2.south) -- ++(0,-0.8);
    \node[cell, fill=nousblue!20] at (2.9, -0.6) {50};

    \draw[arrow, nousgreen!70] (r3.south) -- ++(0,-0.8);
    \node[cell, fill=nousgreen!20] at (4.0, -0.6) {200};

    \draw[arrow, nousgreen!70] (r4.south) -- ++(0,-0.65);
    \node[cell, fill=nousgreen!20] at (5.1, -0.45) {150};

    \draw[arrow, nousorange!80] (r5.south) -- ++(0,-0.35);
    \node[cell, fill=nousorange!20] at (6.2, -0.1) {40};

    \node[font=\tiny, right] at (6.9, -0.15) {$T_0{=}290$};
    \node[font=\tiny, right] at (6.9, -0.5)  {$T_1{=}300$};
    \node[font=\tiny, right] at (6.9, -0.9) {$T_2{=}250$};
\end{scope}

\begin{scope}[yshift=-2.5cm]
    \node[font=\scriptsize\bfseries, anchor=west] at (-0.7, 1.15) {(b) Zigzag Assignment};

    \foreach \i/\tok in {0/300, 1/200, 2/150, 3/100, 4/50, 5/40} {
      \node[cell, fill=nousblue!15] (z\i) at (\i*1.1 + 0.7, 0.6) {$\tau{=}\tok$};
      \node[font=\tiny, above=0pt] at (z\i.north) {$r_{\i}$};
    }

    \node[plabel] at (-0.3, -0.15) {$P_0$};
    \node[plabel] at (-0.3, -0.6) {$P_1$};
    \node[plabel] at (-0.3, -1.05) {$P_2$};

    \draw[arrow, nousblue!70] (z0.south) -- ++(0,-0.35);
    \node[cell, fill=nousgreen!25] at (0.62, -0.1) {300};

    \draw[arrow, nousblue!70] (z1.south) -- ++(0,-0.65);
    \node[cell, fill=nousgreen!25] at (1.8, -0.45) {200};

    \draw[arrow, nousblue!70] (z2.south) -- ++(0,-0.8);
    \node[cell, fill=nousgreen!25] at (2.9, -0.6) {150};

    \draw[arrow, nousorange!80] (z3.south) -- ++(0,-0.8);
    \node[cell, fill=nousorange!25] at (4.0, -0.6) {100};

    \draw[arrow, nousorange!80] (z4.south) -- ++(0,-0.65);
    \node[cell, fill=nousorange!25] at (5.1, -0.45) {50};

    \draw[arrow, nousorange!80] (z5.south) -- ++(0,-0.35);
    \node[cell, fill=nousorange!25] at (6.28, -0.1) {40};

    \node[font=\tiny, right] at (6.9, -0.15) {$T_0{=}340$};
    \node[font=\tiny, right] at (6.9, -0.6)  {$T_1{=}250$};
    \node[font=\tiny, right] at (6.9, -1.05) {$T_2{=}250$};
\end{scope}

\end{tikzpicture}
}
\caption{\small Token-aware partitioning for $P{=}3$ and 6 rows.
\textbf{(a)}~Greedy bin-packing assigns each row to the least-loaded partition; $\lambda{=}1.20$, variable row counts (3,\,1,\,2).
\textbf{(b)}~Zigzag sorts by descending $\tau$, sweeps forward then backward; equal rows ($\lceil 6/3\rceil{=}2$), $\lambda{=}1.36$.}
\label{fig:partitioning-strategies}
\vspace{-10pt}
\end{figure}
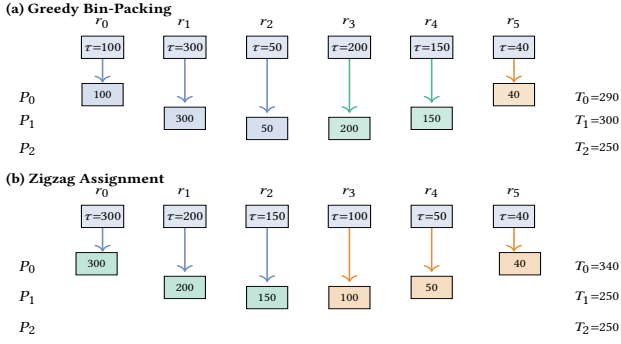


NousDB inserts a repartition node when the token distribution is sufficiently skewed to justify the overhead.
Given a token histogram $H_\tau$, it computes $\text{CV}_\tau = \sigma_\tau/\mu_\tau$ estimating the row-uniform imbalance ratio $\alpha_{\small \text{uniform}} \approx 1 + \text{CV}_\tau \cdot \sqrt{2P\ln P / N}$ via Central Limit Theorem (CLT) bounds on the maximum partition load \cite{10.5555/646975.711521}.
Repartitioning triggers when either (i)~$\alpha_{\text{uniform}} > 1{+}\delta$ ($\delta \in [0.10, 0.15]$), indicating the slowest partition exceeds ideal load by ${\geq}10\%$, or (ii)~the tail ratio $Q_{0.95}/Q_{0.50} > \theta_{\text{tail}}$ (default~$5{\times}$), independently capturing heavy-tail risk where CLT convergence is slow.
When triggered, the default strategy is \textsc{Greedy}.
\textsc{ZigZag} is chosen, when the downstream operator requires
uniform partition cardinality.

\subsection{Query Plan Rewrites} \label{subsec:plan_rewrites}

The algebraic properties in Section \ref{sem_ops} concern individual semantic operators and set operations. We now derive cross-operator equivalences that justify two classes of plan rewrites: \emph{pushdown rules} reduce the input to an expensive operator by moving a cheaper operator below it, while \emph{fusion rules} combine multiple operators into a single model invocation and avoid intermediate materialization.


\begin{property}[Selection Pushdown through Extend]%
    \label{prop:sel-push-extend}
    Given an extend $\varepsilon_{E,A,m,p}$,  and a semantic selection 
    $\sigma_{E',m,p'}$ whose predicate depends
    only on the base attributes $A$. Then:
    \[
      \sigma_{E',m,p'}\!\bigl(\varepsilon_{E,A,m,p}(R)\bigr)
      \;=\;
      \varepsilon_{E,A,m,p}\!\bigl(\sigma_{E',m,p',}(R)\bigr).
    \]
\end{property}
\begin{proof}
    By ~\ref{extend_op},
    $\varepsilon_{E,A,m,p}(R) = \{(r, E(r)) \mid r \in R\}$ where
    $E(r) = m(p, r[A])$. The left-hand side applies $\sigma$ to these extended
    tuples:
    \(
      \bigl\{(r, E(r)) \;\big|\; r \in R,\; \mathit{mp}_{m}(p', (r, E(r))[A])
      = \mathrm{true}\bigr\}.
    \)
    Since $\mathit{mp}$ depends only on $A$ and not on $E$, we have
    $(r, E(r))[A] = r[A]$, so the predicate reduces to
    $\mathit{mp}_{m}(p', r[A]) = \mathrm{true}$. This gives:
    \(
      \bigl\{(r, E(r)) \;\big|\; r \in R,\;
      \mathit{mp}_m(p', r[A]) = \mathrm{true}\bigr\}.
    \)
    On the right-hand side, $\sigma$ first filters $R$ to obtain
    $R' = \{r \in R \mid \mathit{mp}_{m'}(p', r[A]) = \mathrm{true}\}$, and then
    $\varepsilon$ extends $R'$ to produce
    $\{(r, E(r)) \mid r \in R'\} = \{(r, E(r)) \mid r \in R,\;
    \mathit{mp}_{m}(p', r[A]) = \mathrm{true}\}$.
    The two expressions are identical.
\end{proof}



\begin{property}[Semantic Selection Pushdown through Semantic Join]%
    \label{prop:sel-push-join}
    Let $\sigma_{E',m',p'}$ be a Semantic Selection whose predicate
    $\mathit{mp}$ depends only on attributes of $R$.
    Then:
    \[
      \sigma_{E',m',p'}\!\bigl(
        R \bowtie^{\mathit{sem}}_{A,B,m,p} S
      \bigr)
      \;=\;
      \sigma_{E',m',p'}(R)
        \bowtie^{\mathit{sem}}_{A,B,m,p} S.
    \]
\end{property}
\begin{proof}
    We prove the left-pushdown case and the right case is symmetric. 
    Let $P(r) \equiv \operatorname{mp}_m(r[A]) = \text{true}$ as the model-based selection predicate of the semantic selection. 
    By hypothesis, $P(r)$ depends only on the $\sch(R)$ of each output tuple $(r, s)$ of the semantic join, so the selection predicate will be applied on $P\bigl((r,s)\bigr) = P(r)$.
    The left-hand side is:
    \(
      \bigl\{(r, s) \;\big|\; r \in R,\, s \in S,\,
        \theta_m(r[A], s[B]) = \mathrm{true},\,
        P(r) = \mathrm{true}\bigr\}.
    \)
    Since $\theta_m$ and $P(r)$ are independent predicates and conjunction is commutative, this equals:
    \(
      \bigl\{(r, s) \;\big|\; r \in R,\, P(r) = \mathrm{true},\,
        s \in S,\, \theta_m(r[A], s[B]) = \mathrm{true}\bigr\}.
    \)
    The inner condition $\{r \in R \mid P(r)\}$ is precisely $\sigma_{E',m,p'}(R)$, giving $\sigma(R) \bowtie^{\mathit{sem}}_{A,B,m,p} S$.
    Symmetrically, we can prove the same for $S$.
\end{proof}


\begin{property}[Extend--Selection Fusion]\label{prop:ext-sel-fusion}
Let $\varepsilon_{E,A,m,p}(R)$ be an Extend and $\sigma^{\mathit{sem}}_{E,m,p'}$ a Semantic Selection,
both using the same oracle model~$m$, whose predicate depends only on the
extended attributes~$E$, i.e., $\mathit{mp} = \mathit{mp}(E)$.
Assume $E$ is not required by downstream operators (so the final schema is~$A$).
Define a composed prompt~$p_c$ such that for all $r \in R$:
\(
  \mathit{mp}_m(p_c,\, r[A])
  \;=\;
  \mathit{mp}_m\!\bigl(p',\, m(p,\, r[A])\bigr).
\).
Then:
\[
  \pi_A\!\Bigl(
    \sigma^{\mathit{sem}}_{E,\,m,\,p'}\!\bigl(
      \varepsilon_{E,A,m,p}(R)
    \bigr)
  \Bigr)
  \;=\;
  \sigma^{\mathit{sem}}_{A,\,m,\,p_c}(R).
\]
\end{property}
\begin{proof}
    By Definition~\ref{extend_op}, $\varepsilon_{E,A,m,p}(R) = \{(r, E(r)) \mid r \in R,\; E(r) = m(p, r[A])\}$.
    Applying the selection yields $ \mathit{mp}_{m}(p', E(r)) = \mathrm{true}$.
    Projecting to~$A$ gives $\{r \in R \mid \mathit{mp}_{m}(p', E(r)) = \mathrm{true}\}$.
    Since $E(r) = m(p, r[A])$, we have
    $\mathit{mp}_m(p', E(r)) = \mathit{mp}_m(p', m(p, r[A]))$.
    By the definition of~$p_c$, this equals
    $\mathit{mp}_m(p_c, r[A])$ for all~$r$,
    which yields precisely
    $\sigma^{\mathit{sem}}_{A,\,m,\,p_c}(R)$.
\end{proof}


\begin{property}[Extend--Join Fusion]\label{prop:ext-join-fusion}
Let $R$ and $S$ be two input relations. The following three fusion rules
eliminate Extend operators whose output attributes are used solely as join keys
and are not required in the final output. \\
\underline{(a) Left fusion.}\;
Let $E(r) = m(p_e, r[A_R])$ be produced by $\varepsilon_{E,A,m,p_e}(R)$
and used only to join with $S$ on attributes $B$ via prompt $p_j$.
If $\exists$ $p'_j$ s.t. $\forall r \in R$, $s \in S$:
\(
  \theta_{m}\!\bigl(r[A_R],\, s[B]\bigr)
  \;=\;
  \theta_{m}\!\bigl(m(p, r[A_R]),\, s[B]\bigr),
\)
then:
\[
  \pi_{A_R \cup A_S}\!\Bigl(
    \varepsilon_{E,A,m,p_e}(R)
    \bowtie^{\mathit{sem}}_{E,B,m,p_j} S
  \Bigr)
  \;=\;
  R \bowtie^{\mathit{sem}}_{m,B,m,p'_j} S.
\]

\noindent\underline{(b) Right fusion.}\;
Let $E'(s) = m(p'_e, s[A_S])$ be produced by
$\varepsilon_{E',A,m,p'_e}(S)$ and used only as right-side join attributes.
If $\exists$ $p''_j$ s.t. $\forall r \in R$, $s \in S$:
$
  \theta_{m}\!\bigl(r[A],\, s[A_S]\bigr)
  \;=\;
  \theta_{m}\!\bigl(r[A],\, m(p'_e, s[A_S])\bigr),
$
then:
\[
  \pi_{A_R \cup A_S}\!\Bigl(
    R \bowtie^{\mathit{sem}}_{A,E',m,p_j}
    \varepsilon_{E',m,p'_e}(S)
  \Bigr)
  \;=\;
  R \bowtie^{\mathit{sem}}_{A,m,p''_j} S.
\]

\noindent\underline{(c) Bilateral fusion.}\;
    If both input relations are extended ($E(r) = m(p_e, r[A_R])$ and $E'(s) = m(p'_e, s[A_S])$) and both sets of extended attributes are used only as join keys, and $\exists$ a composed  prompt $p^\star_j$ such that for all $r \in R$, $s \in S$:
    \[
      \theta_{m}\!\bigl(r[A_R],\, s[A_S]\bigr)
      \;=\;
      \theta_{m}\!\bigl(m(p_e, r[A_R]),\, m(p'_e, s[A_S])\bigr),
    \]
    then:
    \[
      \pi_{A_R \cup A_S}\!\Bigl(
        \varepsilon_{E,m,p_e}(R)
        \bowtie^{\mathit{sem}}_{E,E',m,p_j}
        \varepsilon_{E',m,p'_e}(S)
      \Bigr)
      \;=\;
      R \bowtie^{\mathit{sem}}_{m,p^\star_j} S.
    \]
    Case (c) follows from applying case (a) and then case (b) in sequence.
\end{property}
\begin{proof}
    We prove case (a); case (b) is symmetric, and case (c) follows by composing the two.
    By Definition~\ref{extend_op}, $\varepsilon_{E,A,m,p_e}(R) = \{(r, E(r)) \mid r \in R,\; E(r) = m(p_e, r[A_R])\}$. 
    By Definition~\ref{definition:join}, the left-hand side is:
    \(
      \bigl\{\bigl((r, E(r)),\, s\bigr) \;\big|\;
        r \in R,\, s \in S,\,
        \theta_{m}\!\bigl(E(r),\, s[B]\bigr) = \mathrm{true}\bigr\}.
    \)
    Projecting to $A_R \cup A_S$ removes the extended attributes $E$, yielding:
    \(
      \bigl\{(r, s) \;\big|\;
        r \in R,\, s \in S,\,
        \theta_{m}\!\bigl(m(p_e, r[A_R]),\, s[B]\bigr) = \mathrm{true}
      \bigr\}.
    \)
    By the hypothesis: $\theta_{m}(r[A_R], s[B]) = \theta_{m}(m(p_e, r[A_R]), s[B])$ for all $r, s$. 
    Substituting gives $R \bowtie^{\mathit{sem}}_{A,B,m,p'_j} S$.
\end{proof}


\noindent\Cref{fig:workflow}, shows how NousDB puts in action all the proposed optimizations.
NousDB applies query rewriting to fuse the two semantic operators and based on input token distribution, it estimates the batch size and whether to insert the token repartition operator.
\section{Experiments} 

\begin{figure}[t]
\centering
\resizebox{\columnwidth}{120pt}{%
\begin{tikzpicture}[
    >=Stealth,
    box/.style={rectangle, rounded corners=2pt, draw=#1, fill=white!1, font=\scriptsize, inner xsep=2.5pt, inner ysep=1.5pt, minimum height=0.38cm, 
    text centered, semithick},
    box/.default={black!40},
    sem/.style={box=violet!60, fill=violet!10},
    optbox/.style={box=orange!60, fill=orange!10, align=center, inner ysep=1pt, minimum width=1.72cm},
    fusedbox/.style={box=orange!70, fill=orange!18, thick},
    neutbox/.style={box=black!30, fill=gray!6},
    execbox/.style={box=red!40, fill=red!6},
    codebox/.style={box=blue!40, fill=blue!6},
    repartb/.style={box=orange!50, fill=orange!8},
    plabel/.style={font=\footnotesize, text=black!70, anchor=north west},
    note/.style={font=\tiny, text=black!50},
    arr/.style={->, >=Stealth, black!45},
    darr/.style={->, >=Stealth, black!30},
    bigarr/.style={->, >=Stealth, thick, black!45},
]
\def\xL{0}     \def\wL{2.85}
\def\xM{3.3}   \def\wM{2.8}
\def\xO{6.55}  \def\wO{1.9}
\def\pH{2.15}
\def\tA{0}      
\def\tB{-2.2}   
\def\cL{1.425} \def\cM{4.6}   \def\cO{7.5}

\begin{scope}[on background layer]
\fill[blue!3, rounded corners=4pt] (\xL,\tA) rectangle ++(\wL,-\pH);
\fill[teal!3, rounded corners=4pt] (\xM,\tA) rectangle ++(\wM,-\pH);
\filldraw[draw=orange!30, fill=orange!3, rounded corners=4pt] (\xO,\tA) rectangle (\xO+\wO,\tB-\pH);
\fill[blue!3, rounded corners=4pt] (\xM,\tB) rectangle ++(\wM,-\pH);
\fill[red!3, rounded corners=4pt] (\xL,\tB) rectangle ++(\wL,-\pH);
\end{scope}

\node[plabel] at (\xL+0.08,\tA-0.06) {1\;\;User program};
\node[codebox, anchor=north, align=left, font=\scriptsize, minimum width=2.5cm] at (\cL,\tA-0.7) (code) {%
    sql("... reviews")\\
    .extend("actor", ...)\\
    .semantic\_join(...)\\
    .select\_columns(...)};

\node[plabel] at (\xM+0.08,\tA-0.06) {2\;\;Logical plan};
\node[neutbox] at (\cM,\tA-0.64) (proj1) {$\pi_{\textit{review, title, name}}$};
\node[sem] at (\cM-0.25,\tA-1.23) (sj1) {$\bowtie^{\mathit{sem}}_{\textit{actor}=\textit{name}}$};
\node[sem] at (\cM-0.5,\tA-1.82) (ext1) {$\varepsilon_{\textit{actor}}(\textit{reviews})$};
\node[neutbox] at (\cM+0.92,\tA-1.82) (nr1) {\textit{name}};
\draw[arr] (sj1) -- (sj1|-proj1.south);
\draw[arr] (ext1) -- (sj1);
\draw[arr] (nr1) -- (sj1);
\node[note, anchor=west, text=red!60, font=\tiny\bfseries, align=left] at ([xshift=0.08cm]sj1.east) {2 LLM\\[-1.5pt]calls/row};

\node[plabel] at (\xO+0.08,\tA-0.06) {3\;\;Optimizer};
\node[note, anchor=north west, align=left] at (\xO+0.1,\tA-0.4) {};
\node[optbox, anchor=north] at (\cO,\tA-0.7) (r1) {%
    Semantic\\[-1.5pt] rewriting\\[-1.5pt]
    {\tiny Extend-join fusion}};
\node[optbox, anchor=north, below=0.45cm of r1] (r2) {%
    Histogram\\[-1.5pt]
    {\tiny $H_\tau$(\textit{review})}};
\node[optbox, anchor=north, below=0.45cm of r2] (r3) {%
    Batch size\\[-1.5pt]
    {\tiny $b_L{=}b_R{=}6$}};
\node[optbox, anchor=north, below=0.45cm of r3] (r4) {%
    Token\\[-1.5pt] repartition\\[-1.5pt]
    {\tiny Greedy bin-pack}};

\draw[darr] (r1) -- (r2);
\draw[darr] (r2) -- (r3);
\draw[darr] (r3) -- (r4);

\node[plabel] at (\xM+0.08,\tB-0.06) {4\;\;Physical plan};
\node[neutbox] at (\cM,\tB-0.64) (proj2) {ProjectionExec};
\node[fusedbox, align=center, inner ysep=1pt] at (\cM-0.25,\tB-1.23) (fj) {SemanticJoinExec\\[-1.5pt]{\tiny fused, $b_L{=}b_R{=}6$}};
\node[repartb, align=center, inner ysep=1pt] at (\cM-0.55,\tB-1.82) (rep2) {TokenRepart\\[-1.5pt]{\tiny bin-pack(reviews)}};
\node[neutbox, align=center, inner ysep=1pt] at (\cM+0.85,\tB-1.82) (nr2) {Parquet\\[-1.5pt]{\tiny (name)}};
\draw[arr] (fj) -- (fj|-proj2.south);
\draw[arr] (rep2) -- (rep2|-fj.south);
\draw[arr] (nr2) -- (fj);
\node[note, anchor=west, text=teal!70, font=\tiny\bfseries, align=left] at ([xshift=0.06cm]fj.east) {1 LLM\\[-1.5pt]call/pair};

\node[plabel] at (\xL+0.08,\tB-0.06) {5\;\;Execution};
\node[note, anchor=north east, align=right] at (\xL+\wL-0.06,\tB-0.07) {};
\node[execbox, minimum width=2.55cm] at (\cL,\tB-0.72) (bnl) {Block nested loop};
\node[execbox, fill=red!10, minimum width=2.55cm] at (\cL,\tB-1.3) (vllm) {vLLM{\tiny Qwen3-8B}};
\node[box=teal!60, fill=teal!8, minimum width=2.55cm] at (\cL,\tB-1.88) (result) {Result {\tiny (review, title, name)}};
\draw[arr] (bnl) -- (vllm);
\draw[arr] (vllm) -- (result);

\draw[bigarr] (\xL+\wL+0.04,\tA-1.0) -- (\xM-0.04,\tA-1.0);
\draw[bigarr] (\xM+\wM+0.04,\tA-1.0) -- (\xO-0.04,\tA-1.0);
\draw[bigarr] (\xO-0.04,\tB-1.0) -- (\xM+\wM+0.04,\tB-1.0);
\draw[bigarr] (\xM-0.04,\tB-1.0) -- (\xL+\wL+0.04,\tB-1.0);
\end{tikzpicture}%
}
\caption{\footnotesize
End-to-end system workflow for query Q11.
 The user program~(\textbf{1}) builds a logical plan~(\textbf{2}) with separate $\varepsilon$ and $\bowtie^{\mathit{sem}}$ operators requiring two LLM calls per row.
The optimizer~(\textbf{3}) applies the new rules.
}
\label{fig:workflow}
\vspace{-10pt}
\end{figure}
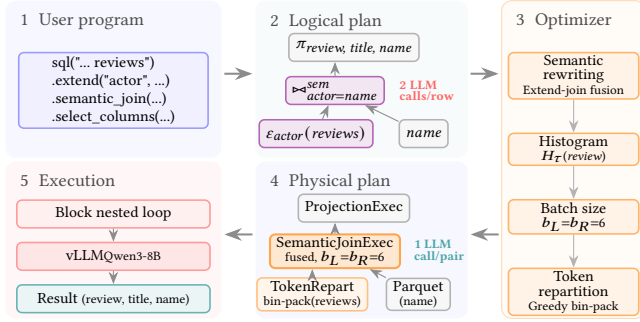

We evaluate NousDB through four research questions: 

\textbf{RQ1: How does NousDB perform end to end?} We compare NousDB with competing systems on IMDB, SemBench~\cite{sembench}, and semantic extensions of JOB~\cite{10.14778/2850583.2850594} and STACK~\cite{10.14778/3503585.3503586} (\Cref{sec:rq1}).

\textbf{RQ2: How do batching and fidelity constraints affect execution?} We measure the token savings from batching and the effects of fidelity constraints on batch size, runtime, and quality (\Cref{sec:rq2}). 

\textbf{RQ3: What is the contribution of each optimization?} We ablate histogram-guided batch sizing, fidelity-aware batching, token-aware repartitioning, and algebraic rewrites (\Cref{sec:rq3}). 

\textbf{RQ4: What are the benefits of the semantic algebra?} We evaluate its effects on query performance and output quality and discuss it across all sections.

 




\subsection{Experimental Setup}\label{sec:setup}

All experiments run on a server equipped with two NVIDIA A6000 GPUs (48\,GB VRAM each), Intel(R) Xeon(R) Gold 5318Y CPU @ 2.10GHz and $377$GB memory.
For model serving we used vLLM \cite{10.1145/3600006.3613165}.
For our experiments we use Qwen3-8B \cite{yang2025qwen3}, a GPT-5.4-nano \cite{singh2025openai} and a Qwen3-32B for the ablation study only.

\textbf{Systems.}
NousDB default setup is histogram-guided batched execution (\textit{Batched}), without model reasoning enabled, and when necessary, we specify when it computes the batch fidelity (\emph{+Certified}) and when it uses model reasoning (\emph{+Thinking}). 
We compare against two state-of-the-art semantic query processing engines:
LOTUS \cite{patel2024semanticoperators} and
Palimpzest \cite{liu2025palimpzest}, a declarative system w/o its  sample-based cost optimizer Abacus \cite{abacus}.
In our analysis we include a vanilla DataFusion \cite{10.1145/3626246.3653368} version, with naive semantic operator implementation, processing their input row-by-row.
All competitors use the same models (Qwen3-8B via vLLM and GPT-5.4-nano via cloud with $32K$ and $400K$ context each) with 10 parallel workers to ensure a fair comparison of system-level optimizations.

\textbf{Workloads.}
We use 4 different benchmarks in our evaluation: \emph{IMDB} extends the JOB schema~\cite{10.14778/2850583.2850594} with 10.5K reviews~\cite{maas-etal-2011-learning}.
It contains 11 queries in total, having six single operator queries (Q1--Q4, and the extend-to-aggregate Q5 and Q6) and the rest are pipelines of several semantic operators (Q7--Q11), four of which the rewrites target (Q7, Q8, Q10, Q11).
\emph{SemBench-Movie}~\cite{sembench} has ten queries: retrieval with a LIMIT (SB1, SB2), aggregates (SB3, SB4, SB8), self-joins of one movie's reviews (SB5 and SB6 with a LIMIT, SB7 exhaustive over 13,924 pairs) and ranking (SB9, SB10).
Finally, we extend JOB ~\cite{10.14778/2850583.2850594} and \emph{StackExchange} benchmarks (24M
rows under the STACK schema~\cite{10.14778/3503585.3503586}, with semantic operators, augmenting standard benchmarks used in the query optimization literature.

\subsection{Comparison against Workloads} \label{sec:rq1}

\begin{figure}[t]
  \centering
  \includegraphics[width=\columnwidth]{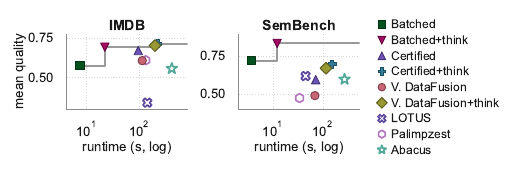}
    \vspace{-30pt}
  \caption{\small Mean quality against runtime with GPT-5.4-nano: the six NousDB variants and competitors on the Pareto frontier.}
  
  \label{fig:frontier}
  \vspace{-12pt}
\end{figure}

\textbf{Quality against runtime.}
With GPT-5.4-nano (\Cref{fig:frontier}) only NousDB variants form the Pareto frontier: \emph{Batched} is the fastest point (7.4s on IMDB, 3.6s on SemBench), and \emph{Batched+think}, achieving 0.69 at 22s and 0.84 at 12s, dominates every external system, whose best mean quality is Palimpzest's 0.61 at 132s on IMDB and LOTUS's 0.62 at 45s on SemBench.
Similarly with Qwen3-8B, the IMDB frontier is again NousDB's alone, from \emph{Batched} (0.58 at 22\,s) to \emph{Certified+think} (0.72), while on SemBench LOTUS joins it (0.76 at 12s, between \emph{Batched's} 0.70 at 6.2\,s and V.DataFusion+think's 0.77 at 638\,s), the competitor that no NousDB configuration dominates.
From our experiments, reasoning multiplies the decode cost of every call, since 90--97\% of the thinking arms' output tokens on Qwen are reasoning tokens, and packing $k$ rows per call amortizes exactly that cost, leading to a \emph{if you think, batch} rule of thumb.

\begin{table}[t]
    \centering
    \resizebox
    {\columnwidth}
    {!}
    {
    \begin{tabular}{@{}l@{\hspace{3pt}}rrrrr@{\hspace{6pt}}rrrrr@{}}
    \toprule
     & \multicolumn{5}{c}{IMDB (11 queries)} & \multicolumn{5}{c@{}}{SemBench (10 queries)} \\
    \cmidrule(lr){2-6}\cmidrule(l){7-11}
     & F1 & Qual. & Wall\,(s) & Tok.\,M & $\times$ & F1 & Qual. & Wall\,(s) & Tok.\,M & $\times$ \\
    \midrule
    \multicolumn{11}{@{}l}{\textit{Qwen3-8B (sf 1000)}} \\
    \textbf{Batched} & 0.54 & 0.58 & 22 & 3.0 & 1 & 0.62 & 0.70 & 6.2 & 1.1 & 1 \\
    Certified & 0.56 & 0.59 & 83 & 7.4 & 3.79 & 0.47 & 0.64 & 27 & 28.9 & 4.42 \\
    Row-by-row+think & 0.64 & 0.67 & 980 & 18.9 & 44.6 & 0.77 & \textbf{0.77} & 638 & 61.6 & 103.3 \\
    Batched+think & 0.67 & 0.71 & 151 & 4.2 & 6.90 & 0.78 & 0.75 & 90 & 5.1 & 14.6 \\
    Certified+think & \textbf{0.69} & \textbf{0.72} & 1204 & 10.3 & 54.8 & \textbf{0.79} & 0.74 & 924 & 53.8 & 149.6 \\
    \addlinespace[1pt]
    V.DataFusion & 0.48 & 0.52 & 24 & 9.5 & 1.08 & 0.47 & 0.62 & 16 & 29.0 & 2.59 \\
    LOTUS & 0.28 & 0.37 & 42 & 4.7 & 1.91 & 0.75 & 0.76 & 12 & 5.7 & 1.94 \\
    Palimpzest & 0.67 & 0.69 & 221 & 8.7 & 10.1 & 0.65 & 0.61 & 64 & 9.5 & 10.4 \\
    Abacus & 0.61 & 0.67 & 837 & 2.7 & 38.1 & 0.60 & 0.67 & 337 & 0.45 & 72.6 \\
    \midrule
    \multicolumn{11}{@{}l}{\textit{GPT-5.4-nano (sf 1000)}} \\
    Row-by-row$^{a}$ & 0.56 & 0.61 & 113 & 10.2 & 15.2 & 0.48 & 0.49 & 68 & 31.8 & 18.9 \\
    \textbf{Batched} & 0.50 & 0.58 & 7.4 & 2.6 & 1 & 0.63 & 0.72 & 3.6 & 1.8 & 1 \\
    Certified & 0.64 & 0.67 & 95 & 5.8 & 12.8 & 0.50 & 0.60 & 71 & 31.6 & 19.7 \\
    Row-by-row+think & 0.68 & 0.70 & 197 & 17.1 & 26.6 & 0.78 & 0.67 & 114 & 53.8 & 31.6 \\
    Batched+think & 0.66 & 0.69 & 22 & 3.0 & 2.97 & \textbf{0.85} & \textbf{0.84} & 12 & 2.9 & 3.31 \\
    Certified+think & \textbf{0.69} & \textbf{0.72} & 221 & 7.3 & 29.8 & 0.79 & 0.70 & 152 & 53.8 & 41.9 \\
    \addlinespace[1pt]
    V.DataFusion & 0.48 & 0.52 & 24 & 9.5 & 1.08 & 0.47 & 0.62 & 16 & 29.0 & 2.59 \\
    LOTUS & 0.22 & 0.34 & 141 & 4.3 & 19.3 & 0.66 & 0.62 & 45 & 3.1 & 12.3 \\
    Palimpzest & 0.57 & 0.61 & 131 & 8.1 & 17.8 & 0.62 & 0.47 & 34 & 8.0 & 9.32 \\
    Abacus & 0.46 & 0.56 & 411 & 14.0 & 59.3 & 0.66 & 0.60 & 269 & 1.4 & 86.3 \\
    \bottomrule
    \end{tabular}
    }
    \caption{\small NousDB variants and competitors: mean F1, mean normalized quality over all queries (Qual.), geometric-mean wall clock (s), total tokens (M), and $\times$ = paired geometric-mean speedup of NousDB}
    \vspace{-20pt}
    \label{tab:all_systems}
\end{table}

\noindent\textbf{Runtime.}
\emph{Batched} NousDB is the fastest system in every benchmark and model setting (\Cref{fig:latency-grid}).
With GPT-5.4-nano it is $15$--$19\times$ faster than V.~DataFusion, LOTUS and Palimpzest on IMDB and $9$--$19\times$ on SemBench, and $59$--$86\times$ faster than Abacus, whose plan search dominates its runtime.
On Qwen3-8B it is $1.9\times$ faster than LOTUS, $10\times$ than Palimpzest and $38$--$73\times$ than Abacus, but only $1.08\times$ faster than V.~DataFusion on IMDB, where both run their semantic operators on ten partitions and vLLM's continuous batching and prefix cache already make per-row requests cheap.
On SemBench, the movie filter leaves only 77--118 rows on five queries, which NousDB, V.~DataFusion and LOTUS finish within $1.5\times$ of one another while Palimpzest and Abacus pay 10--66\,s and 205--400\,s of fixed overhead.
On the LIMIT-bounded SB1, SB5 and SB6, NousDB answers from a bounded prefix of its input (Ref section 1 of the experimental), $2.5$--$38\times$ faster than V.~DataFusion and LOTUS.
Doubling the input (dots in \Cref{fig:latency-grid}) grows NousDB's runtime by $1.26$--$1.92\times$ and that of V.~DataFusion, LOTUS and Palimpzest by $1.7$--$2.2\times$ (geometric means per benchmark and model), so 14 of the 16 speedups grow, up to $85\times$ on IMDB and $82\times$ on SemBench, while the two that shrink are against Abacus, whose fixed-budget plan search takes about as long at either scale.

\textbf{Quality and Total Tokens.}
Focusing on the token consumption reported in \Cref{tab:all_systems}, batching and plan rewrites make NousDB having the lowest token consumption in total on 40 out of 42 queries.
In terms of quality, no system is dominant on every query, with Qwen3-8B Palimpzest and Abacus lead IMDB and LOTUS on SemBench, paying the overhead of $18-20\times$ of NousDB's runtime.
    
NousDB achieves Pareto frontier optimality primarily through its implementation of the proposed extended algebra, especially the \emph{pushdowns} and \emph{fusions}.
Pushing Q7's selection below its Extend operator, reduces the tuples reaching the model from $2000$ to $1606$ on Qwen3 and to $1166$ to GPT, achieving a $20\%$ and $42\%$ reduction on token consumption respectively, while the application of semantic filter before the semantic join on Q8, cuts the $44-47\%$ of the tokens to reach the second semantic opertor.
In all the cases above the F1 score varied only $0.01$, compared to full execution of the plan.
On Q10 the fusion of the Extend-selection halves the total rows sent to the model, while improving the F1 score from $0.24$ to $0.47$ and from $0.13$ to $0.57$ for Qwen and GPT respectively.
Finally the fused plan on Q11, managed to achieve the best F1 score across all systems.

\begin{figure}[t]
      \vspace{-10pt}
  \centering
  \includegraphics[width=\columnwidth]{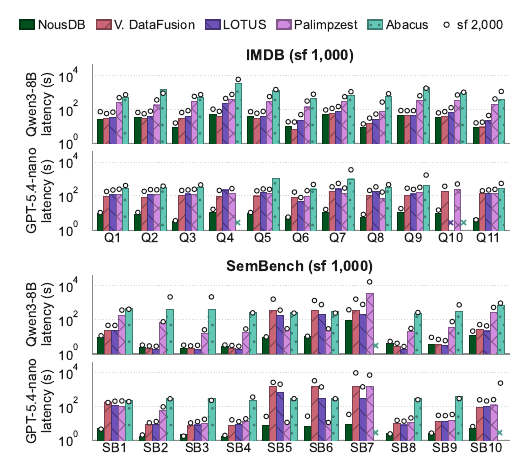}

  \caption{\small End-to-end latency per query (log scale) bars indicate at sf 1$,$000, dots at sf 2$,$000}
  \label{fig:latency-grid}
  \vspace{-15pt}
\end{figure}

\textbf{JOB and STACK workloads.}
We augment two well-known benchmarks for query optimization with semantic operators, keeping each source query's join graph, predicates and outputs and adding semantic operators on a text attribute its joins reach: JOB~\cite{10.14778/2850583.2850594} gets the \textsf{reviews} relation, joined to \textsf{title}, and STACK~\cite{10.14778/3503585.3503586} already exposes question, answer and comment bodies.
Each query gets one semantic configuration (a selection, a one- or multi-attribute Extend, a selection--Extend pipeline, or a semantic join), yielding 33 JOB queries and 10 STACK queries.
\Cref{fig:job-stack} contains the comparison between the systems on these two benchmarks.
In 19 of them the original join predicates drive the intermediate result to zero cardinality before the semantic operator fires, where on this \emph{relational-dominant} subset the three systems fall within $1.2\times$ of one another, showing that semantic augmentation adds negligible overhead when the operator is inactive.
The remaining 14 \emph{semantic-active} queries trigger between 3 and 1{,}265 LLM calls.
On this subset NousDB achieves a geometric-mean speedup of $3.3\times$ over LOTUS and $2.8\times$ over Palimpzest, winning 13 of 14 and 12 of 14 respectively.
The advantage peaks on complex-tier queries whose join graphs emit ${\sim}500$K intermediate rows: on J8, NousDB finishes in 717\,s versus 25{,}038\,s (LOTUS, $35\times$) and 22{,}204\,s (Palimpzest, $31\times$), and consumes 16.1\,M tokens versus 29.9\,M ($1.86\times$ fewer), as the shared instruction prefix is amortised across the batch.
The regression on J24, where NousDB trails LOTUS by $1.3\times$ due to a join topology that limits the greedy bin-packing operator.

In STACK benchmark, we evaluate NousDB at $\text{sf}\!=\!10{,}000$ while LOTUS and Palimpzest run at $\text{sf}\!=\!1{,}000$.
Despite $10\times$ more data, NousDB achieves geometric-mean latency $4.9\times$ lower than LOTUS and $5.3\times$ lower than Palimpzest, winning 8 of 10 and 9 of 10 pairwise comparisons.
The two near-parity queries (S1, S10) are relational-dominant, each triggering at most two LLM calls.
On the most token-intensive query (S11), NousDB processes 7{,}146 rows and 2.37\,M tokens in 84\,s, versus 222\,s (LOTUS) and 268\,s (Palimpzest).

\begin{figure}[t]
  \centering
  \includegraphics[width=\columnwidth]{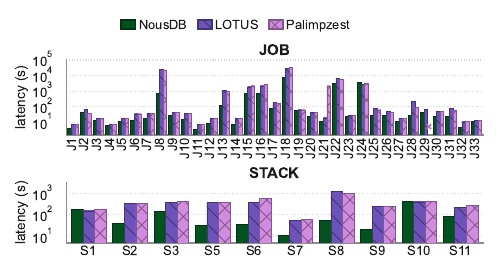}
  \caption{\small JOB (top) and STACK (bottom). NousDB at sf 10,000, rest at sf 1,000) latency per query (log scale).}
  \label{fig:job-stack}
  \vspace{-15pt}
\end{figure}

\subsection{Batch Sizing and Batch Fidelity}\label{sec:rq2}

\begin{figure}[t]
  \centering
  \includegraphics[width=0.9\columnwidth]{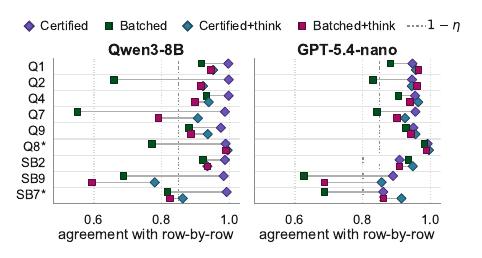}
  \vspace{-18pt}
  \caption{\small Share of the input tuples on which Certified and Batched execution return the singleton answer. Batch sizes as in \Cref{fig:certified-k}.}
  \label{fig:invariance}
  \vspace{-10pt}
\end{figure}

\begin{figure}[t]
  \centering
  \begin{subfigure}{\columnwidth}
    \centering
    \includegraphics[width=\columnwidth]{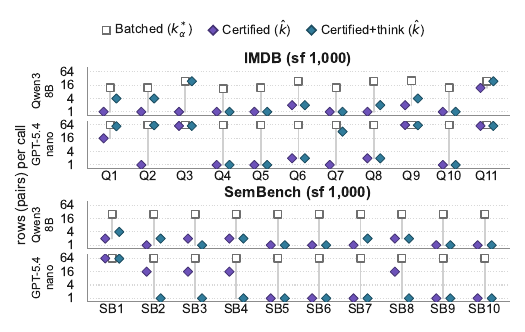}
    \vspace{-15pt}
    \caption{\small Rows per call of Batched ($k_{H_\tau}$) and of the certificate without and with thinking ($k_{safe}$)}
    \label{fig:certified-k}
  \end{subfigure}\\[3pt]
  \begin{subfigure}{\columnwidth}
    \centering
    \includegraphics[width=\columnwidth]{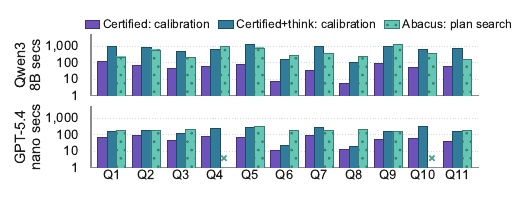}
    \vspace{-20pt}
    \caption{\small Seconds per IMDB query in the certificate's calibration prologue, without and with thinking, and in Abacus's plan search.}
    \label{fig:quality-overhead}
  \end{subfigure}
  \vspace{-5pt}
  \caption{\small What the fidelity certificate chooses and what it costs, per query at sf 1,000.}
  \label{fig:certificate}
  \vspace{-10pt}
\end{figure}

\textbf{Impact on token consumption.}
Histogram-guided batching alone (C1 in \Cref{tab:ablation}) cuts model calls to $\times0.05$--$0.17$ and tokens to $\times0.15$--$0.32$ of per-row execution on every model and benchmark.
The saving is on the input side: packing $k$ rows shares the instruction prefix while still the model generates an answer per input row, so the joins, whose per-pair prompts repeat the instructions for every pair, save most (Q3 with Qwen3-8B: 880\,K to 91\,K tokens).

\textbf{Impact on quality.}
On most queries batching has a small impact o overall quality: over the 36 pairs of \Cref{fig:invariance}, a batched variant departs from its per-row counterpart on 15\% of the tuples on average (median 10\%), with one extreme case, Q7 with Qwen3-8B (45\%).
Qwen3-8B is the more batch-sensitive model (17\% against 12\% for GPT-5.4-nano), and its quality degredation comes earlier, as shown in  \Cref{fig:distributions_and_sweep}.
The above behavior is regulated on both models when their reasoning capabilities are enabled.

\textbf{What the certificate guarantees and costs.}
Certified execution returns the per-row answer on 97.5--99.8\% of the tuples with Qwen3-8B and 86--99\% with GPT-5.4-nano, above $1-\eta$ in every cell, including where the predicted batch sizes where large ($k_{safe}=16$--64 with GPT-5.4-nano: 91--95\%).
At sf 2,000, \emph{Batched} falls below $1-\eta$ on 8 of 16 cells and \emph{Certified} on none.
With Qwen3-8B the certificate returns $k_{safe} \leq$ 2 on 10 of 11 and on IMDB queries and on every SemBench query.
GPT-5.4-nano, is more resistant to batching as a bigger model, where it keeps $k_{safe}$ between 16 and 64 on four queries (\Cref{fig:certified-k}).
The calibration prologue takes 6--118\,s (Qwen3-8B) and 11--97\,s (GPT-5.4-nano) per IMDB query, about half of a certified run, and does not grow with the input, whereas Abacus's plan search takes 157--1,259\,s (\Cref{fig:quality-overhead}).

\begin{figure}[t]
  \vspace{-2pt}
  \centering
  \includegraphics[width=\columnwidth]{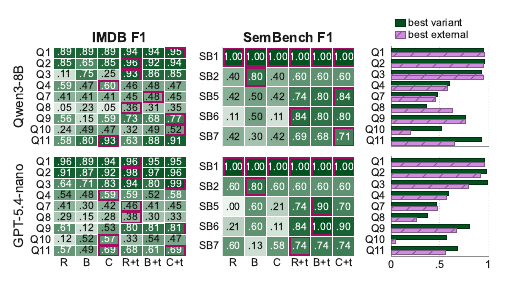}
  \caption{\small F1 per query at sf 1,000 of the six NousDB, per IMDB query and the best variant against the best of LOTUS, Palimpzest and Abacus (right).}
  \label{fig:arm-quality-map}
  \vspace{-15pt}
\end{figure}

\textbf{Quality of the variants.}
The choice of variant moves F1 as much as the choice of system (\Cref{fig:arm-quality-map}): on Q3 the variants span 0.11 \emph{(Row-by-row)} to 0.93 \emph{(Row-by-row+think)} with Qwen3-8B, and on Q9 0.12 \emph{(Batched)} to 0.81 \emph{(Certified+think)} with GPT-5.4-nano.
The best variant matches or beats the best external system on 6 of 9 IMDB queries with Qwen3-8B and 7 of 9 with GPT-5.4-nano, where it concedes only Q1 and Q7 to Palimpzest, by 0.004 and 0.016.
The one large gap is Q8 with Qwen3-8B (0.36 against Palimpzest's 0.63), which GPT-5.4-nano's best variant wins by 0.11.
Finally on SemBench, \emph{Batched} is the best variant on SB2 (0.80) with both models, and the thinking variants take the pair queries \emph{(Batched+think}: 0.90 on SB5 and 1.00 on SB6 with GPT-5.4-nano).


\subsection{Ablation Study}\label{sec:rq3}

\begin{table}[t]
  \centering
  \caption{\footnotesize Optimizations benefits. Top: geometric-mean ratio to C0 over each benchmark's ablation queries (below 1 = better). Bottom: answer quality of every cell (F1; SB3: relative error, lower is better; SB10: Spearman $\rho$; bold = best of the six configurations per model).}
  \label{tab:ablation}
  \vspace{-8pt}
  \scriptsize
  \setlength{\tabcolsep}{0pt}%
  \def\cc#1{{\setbox0\hbox{.000}\hbox to\wd0{\hss#1\hss}}}%
  \begin{tabular*}{\columnwidth}{@{}l@{\extracolsep{\fill}}*{6}{r}@{\hspace{2pt}}*{6}{r}@{\hspace{2pt}}*{6}{r}@{}}
  \toprule
   & \multicolumn{6}{c}{Qwen3-8B} & \multicolumn{6}{c}{Qwen3-32B} & \multicolumn{6}{c}{GPT-5.4-nano} \\
  \cmidrule(l{2pt}r{2pt}){2-7}\cmidrule(l{2pt}r{2pt}){8-13}\cmidrule(l{2pt}){14-19}
   & \cc{C0} & \cc{C1} & \cc{C2} & \cc{C3} & \cc{C4} & \cc{C5} & \cc{C0} & \cc{C1} & \cc{C2} & \cc{C3} & \cc{C4} & \cc{C5} & \cc{C0} & \cc{C1} & \cc{C2} & \cc{C3} & \cc{C4} & \cc{C5} \\
  \midrule
  \multicolumn{19}{@{}l}{\textit{IMDB: ratio to C0 (geometric mean)}} \\
  calls &  & 0.08 & 1.00 & 0.51 & 0.11 & 0.07 &  & 0.07 & 1.00 & 0.49 & 0.13 & 0.08 &  & 0.05 & 1.01 & 0.91 & 0.11 & 0.05 \\
  tokens &  & 0.26 & 1.00 & 0.68 & 0.35 & 0.26 &  & 0.26 & 1.00 & 0.65 & 0.37 & 0.29 &  & 0.29 & 1.00 & 0.99 & 0.44 & 0.25 \\
  wall &  & 1.39 & 0.59 & 0.69 & 1.08 & 0.74 &  & 1.22 & 0.95 & 0.84 & 0.80 & 0.57 &  & 0.25 & 0.24 & 1.06 & 0.24 & 0.18 \\
  \addlinespace[2pt]
  \multicolumn{19}{@{}l}{\textit{SemBench: ratio to C0 (geometric mean)}} \\
  calls &  & 0.14 & 1.03 & 0.89 & 0.16 & 0.17 &  & 0.17 & 0.84 & 0.93 & 0.18 & 0.18 &  & 0.09 & 1.04 & 1.00 & 0.09 & 0.09 \\
  tokens &  & 0.15 & 1.00 & 1.00 & 0.17 & 0.17 &  & 0.17 & 1.00 & 1.00 & 0.18 & 0.18 &  & 0.32 & 1.06 & 1.00 & 0.41 & 0.41 \\
  wall &  & 4.76 & 0.90 & 0.89 & 2.05 & 1.52 &  & 3.77 & 0.97 & 0.96 & 1.49 & 1.49 &  & 0.49 & 0.33 & 0.94 & 0.30 & 0.31 \\
  \midrule
  \multicolumn{19}{@{}l}{\textit{Answer quality}} \\
  Q1 & \cc{\textbf{.892}} & .883 & \cc{\textbf{.892}} & \cc{\textbf{.892}} & .888 & .890 & .917 & .922 & .916 & .916 & \cc{\textbf{.930}} & \cc{\textbf{.930}} & .927 & \cc{\textbf{.931}} & \cc{\textbf{.931}} & .930 & .895 & .892 \\
  Q3* & .114 & .484 & .114 & .396 & .370 & \cc{\textbf{.759}} & .776 & \cc{\textbf{.846}} & .783 & .792 & .796 & .818 & .677 & .537 & .608 & \cc{\textbf{.900}} & .800 & .718 \\
  Q4 & \cc{\textbf{.585}} & .438 & \cc{\textbf{.585}} & \cc{\textbf{.585}} & .456 & .456 & .595 & .583 & .595 & \cc{\textbf{.607}} & .513 & .513 & .500 & .538 & \cc{\textbf{.582}} & .537 & .526 & .500 \\
  Q7 & .411 & \cc{\textbf{.433}} & .413 & .411 & .414 & .410 & .460 & .429 & .462 & .464 & .455 & \cc{\textbf{.486}} & .368 & .292 & \cc{\textbf{.369}} & .366 & .226 & .316 \\
  Q11* & .582 & .421 & .582 & .611 & .435 & \cc{\textbf{.789}} & .686 & .568 & .686 & \cc{\textbf{.909}} & .495 & .835 & .525 & .421 & .571 & \cc{\textbf{.819}} & .500 & .467 \\
  \addlinespace[2pt]
  SB7* & \cc{\textbf{.428}} & .290 & \cc{\textbf{.428}} & \cc{\textbf{.428}} & .327 & .328 & \cc{\textbf{.456}} & .394 & .455 & .455 & .389 & .381 & .591 & .180 & \cc{\textbf{.604}} & .601 & .169 & .216 \\
  SB3 & \cc{\textbf{.111}} & .556 & \cc{\textbf{.111}} & \cc{\textbf{.111}} & .222 & .222 & .556 & .667 & .556 & .556 & \cc{\textbf{.111}} & \cc{\textbf{.111}} & .667 & .444 & .556 & .556 & \cc{\textbf{.111}} & \cc{\textbf{.111}} \\
  SB10 & .286 & .202 & .286 & .286 & .335 & \cc{\textbf{.349}} & .329 & .311 & .333 & .329 & .343 & \cc{\textbf{.352}} & .303 & .319 & .301 & .289 & .317 & \cc{\textbf{.358}} \\
  \bottomrule
  \end{tabular*}
  \vspace{-12pt}
\end{table}

We isolate the three optimizations on queries that cover every operator type on IMDB and SemBench at sf 1,000, adding also Qwen3-32B.
C0 is per-row execution, C1, C2 and C3 each enable one optimization from batching, token repartitioning and rewrites respectively, C4 = C1+C2, and C5 is full NousDB (\Cref{tab:ablation}).

\textbf{Batch sizing leads to significant cost reduction (C1).}
Batching alone cuts calls to $\times0.05$--$0.17$ and tokens to $\times0.15$--$0.32$ of C0 on every model and benchmark.
With GPT-5.4-nano the wall clock follows ($\times0.25$ on IMDB, $\times0.49$ on SemBench): the join Q3 packs 986 per-pair calls into 34, taking 62\,K tokens instead of 944\,K and 8\,s instead of 122\,s.
The self-hosted models however, get slower ($\times1.22$--$4.76$): vLLM's continuous batching already batches concurrent per-row requests, so packing mainly lengthens each request (Q4 with Qwen3-8B: 204\,s against 97\,s on the same ten partitions).

\textbf{Token repartitioning as a new shuffling technique (C2).}
Token Repartition shuffles rows, so in terms of tokens and calls is identical to C0.
Its gain is the granularity of the shuffle:
DataFusion’s round-robin partitioning operator shuffles record batches across partitions, while token repartitioning shuffles them based on their tokens.
GPT-5.4 shows big gains when the repartitioning is enables, while the exchange already dealt single rows to ten partitions (session batch size of one), and C2 stays within 6\% of C0: balance has little to fix.
The true gains of this operator come hand-to-hand when paired with batching, as we observe in C5.

\textbf{Rewrites change what is computed (C3).}
From query rewrites only IMDB is affected, since only Q7 applies filter pushdown, and Q11 fusion of semantic operators, leading to  20--36\% fewer calls, preserving or boosting F1 score.
The rest of the queries remain unchanged, leading to overall gains in all observed axes, hinting that the extended algebra and its properties play a significant role when integrated with traditional query engines.

\textbf{Repartitioning smooths batches (C4).}
On top of batching, repartitioning removes most of C1's wall-clock penalty on the self-hosted models (with Qwen3-8B, $\times1.39\to1.08$ on IMDB and $\times4.76\to2.05$ on SemBench) but returns part of the token consumption savings (IMDB: $\times0.35$--$0.44$, from $\times0.26$--$0.29$).
The operator deflate the partitions' token loads ($\lambda\le1.14$ on row inputs) and this s spread is what produces: $\times0.33$--$0.84$ speedups on the 12 Qwen pairs.
Thus, repartitioning operator is valuable only when the engine is capable of processing rows in batches, and 

\textbf{Putting all together (C5).}
Full NousDB contains all optimizations leading to the following reductions: $\times0.05$--$0.18$ on model calls, $\times0.17$--$0.41$ on token consumption and $\times0.18$--$0.31$ at wall clock with GPT-5.4-nano and $\times0.57$--$0.74$ on IMDB with Qwen.
On SemBench with Qwen, stays slower than C0 ($\times1.5$) through SB1 and SB3, which C0 finishes in 2--12\,s, but in this setting SB1 has a LIMIT operator on the query and the whole benchmark consists of single semantic operator queries, making the rewrites never applied.
On the joins Q3 and Q11, Qwen3-8B reaches its best F1 of the six configurations (0.759 and 0.789, from 0.114 and 0.582 in C0) at $\times0.10$ and $\times0.15$ of C0's tokens.
However, C5 is not the quality optimum everywhere: in 64-pair blocks GPT-5.4-nano's joins fall to 0.467 on Q11 and 0.216 on SB7, from 0.819 (C3) and 0.591 (C0), which is what the certificate guards against.

\subsection{Summary of Findings}\label{sec:findings}
From our evaluation we extracted the following findings:
\textbf{RQ1:} NousDB is the fastest system in every benchmark and model setting (up to $86\times$ and  $1.70$--$5.25\times$ in geometric mean on JOB and STACK), spends the fewest tokens on 40 of 42 query cells, and a NousDB configuration is the most accurate in each.
\textbf{RQ2:} Batch sizing cuts tokens to $\times0.15$--$0.32$ of per-row execution without a context overflow, but changes up to 45\% of the answers.
The certificate keeps at least 86\% of them, and reasoning, affordable only when batched, raises quality.
\textbf{RQ3:} Batch sizing carries the cost reduction, rewriting reduces the work itself, repartitioning lowers latency by spreading requests over all partitions, and the three compose.
All optimizations together smoothen the weaknesses caused from each one separately, making NousDB optimizations work cooperatively rather than competitively.
\textbf{RQ4:} Pushdowns remove 20--47\% of the tuples that reach a model and preserve the answers up to run-to-run noise, while fusion improves them (Q10's F1 doubles with Qwen3-8B and quadruples with GPT-5.4-nano).

\section{Conclusions}
We presented an extended algebra for semantic query processing by formalizing semantic operators that compose with relational algebra, enabling semantic query rewriting and a system that implements it, NousDB. 
We showed that semantic analytics fundamentally departs from classical optimization by shifting the unit of cost from rows to tokens, which introduces new challenges in batch sizing and parallel execution under token skew.
We addressed these challenges with a token-based cost model, statistical guarantees for safe and accurate batching, and token-aware repartitioning, integrated into a modern query engine. 
Future work includes extending the algebra with additional semantic operators, incorporating quality into cost models, and query optimization.




\bibliographystyle{ACM-Reference-Format}
\bibliography{refs}

@String{Computing = "Computing" }

@String{Computer = "{IEEE} Computer" }

@String{Chelsea = "Chelsea" }

@misc{openai2024gpt4technicalreport,
      title={GPT-4 Technical Report}, 
      author={OpenAI and Josh Achiam and Steven Adler and Sandhini Agarwal and Lama Ahmad and Ilge Akkaya and Florencia Leoni Aleman and Diogo Almeida and Janko Altenschmidt and Sam Altman and Shyamal Anadkat and Red Avila and Igor Babuschkin and Suchir Balaji and Valerie Balcom and Paul Baltescu and Haiming Bao and Mohammad Bavarian and Jeff Belgum and Irwan Bello and Jake Berdine and Gabriel Bernadett-Shapiro and Christopher Berner and Lenny Bogdonoff and Oleg Boiko and Madelaine Boyd and Anna-Luisa Brakman and Greg Brockman and Tim Brooks and Miles Brundage and Kevin Button and Trevor Cai and Rosie Campbell and Andrew Cann and Brittany Carey and Chelsea Carlson and Rory Carmichael and Brooke Chan and Che Chang and Fotis Chantzis and Derek Chen and Sully Chen and Ruby Chen and Jason Chen and Mark Chen and Ben Chess and Chester Cho and Casey Chu and Hyung Won Chung and Dave Cummings and Jeremiah Currier and Yunxing Dai and Cory Decareaux and Thomas Degry and Noah Deutsch and Damien Deville and Arka Dhar and David Dohan and Steve Dowling and Sheila Dunning and Adrien Ecoffet and Atty Eleti and Tyna Eloundou and David Farhi and Liam Fedus and Niko Felix and Simón Posada Fishman and Juston Forte and Isabella Fulford and Leo Gao and Elie Georges and Christian Gibson and Vik Goel and Tarun Gogineni and Gabriel Goh and Rapha Gontijo-Lopes and Jonathan Gordon and Morgan Grafstein and Scott Gray and Ryan Greene and Joshua Gross and Shixiang Shane Gu and Yufei Guo and Chris Hallacy and Jesse Han and Jeff Harris and Yuchen He and Mike Heaton and Johannes Heidecke and Chris Hesse and Alan Hickey and Wade Hickey and Peter Hoeschele and Brandon Houghton and Kenny Hsu and Shengli Hu and Xin Hu and Joost Huizinga and Shantanu Jain and Shawn Jain and Joanne Jang and Angela Jiang and Roger Jiang and Haozhun Jin and Denny Jin and Shino Jomoto and Billie Jonn and Heewoo Jun and Tomer Kaftan and Łukasz Kaiser and Ali Kamali and Ingmar Kanitscheider and Nitish Shirish Keskar and Tabarak Khan and Logan Kilpatrick and Jong Wook Kim and Christina Kim and Yongjik Kim and Jan Hendrik Kirchner and Jamie Kiros and Matt Knight and Daniel Kokotajlo and Łukasz Kondraciuk and Andrew Kondrich and Aris Konstantinidis and Kyle Kosic and Gretchen Krueger and Vishal Kuo and Michael Lampe and Ikai Lan and Teddy Lee and Jan Leike and Jade Leung and Daniel Levy and Chak Ming Li and Rachel Lim and Molly Lin and Stephanie Lin and Mateusz Litwin and Theresa Lopez and Ryan Lowe and Patricia Lue and Anna Makanju and Kim Malfacini and Sam Manning and Todor Markov and Yaniv Markovski and Bianca Martin and Katie Mayer and Andrew Mayne and Bob McGrew and Scott Mayer McKinney and Christine McLeavey and Paul McMillan and Jake McNeil and David Medina and Aalok Mehta and Jacob Menick and Luke Metz and Andrey Mishchenko and Pamela Mishkin and Vinnie Monaco and Evan Morikawa and Daniel Mossing and Tong Mu and Mira Murati and Oleg Murk and David Mély and Ashvin Nair and Reiichiro Nakano and Rajeev Nayak and Arvind Neelakantan and Richard Ngo and Hyeonwoo Noh and Long Ouyang and Cullen O'Keefe and Jakub Pachocki and Alex Paino and Joe Palermo and Ashley Pantuliano and Giambattista Parascandolo and Joel Parish and Emy Parparita and Alex Passos and Mikhail Pavlov and Andrew Peng and Adam Perelman and Filipe de Avila Belbute Peres and Michael Petrov and Henrique Ponde de Oliveira Pinto and Michael and Pokorny and Michelle Pokrass and Vitchyr H. Pong and Tolly Powell and Alethea Power and Boris Power and Elizabeth Proehl and Raul Puri and Alec Radford and Jack Rae and Aditya Ramesh and Cameron Raymond and Francis Real and Kendra Rimbach and Carl Ross and Bob Rotsted and Henri Roussez and Nick Ryder and Mario Saltarelli and Ted Sanders and Shibani Santurkar and Girish Sastry and Heather Schmidt and David Schnurr and John Schulman and Daniel Selsam and Kyla Sheppard and Toki Sherbakov and Jessica Shieh and Sarah Shoker and Pranav Shyam and Szymon Sidor and Eric Sigler and Maddie Simens and Jordan Sitkin and Katarina Slama and Ian Sohl and Benjamin Sokolowsky and Yang Song and Natalie Staudacher and Felipe Petroski Such and Natalie Summers and Ilya Sutskever and Jie Tang and Nikolas Tezak and Madeleine B. Thompson and Phil Tillet and Amin Tootoonchian and Elizabeth Tseng and Preston Tuggle and Nick Turley and Jerry Tworek and Juan Felipe Cerón Uribe and Andrea Vallone and Arun Vijayvergiya and Chelsea Voss and Carroll Wainwright and Justin Jay Wang and Alvin Wang and Ben Wang and Jonathan Ward and Jason Wei and CJ Weinmann and Akila Welihinda and Peter Welinder and Jiayi Weng and Lilian Weng and Matt Wiethoff and Dave Willner and Clemens Winter and Samuel Wolrich and Hannah Wong and Lauren Workman and Sherwin Wu and Jeff Wu and Michael Wu and Kai Xiao and Tao Xu and Sarah Yoo and Kevin Yu and Qiming Yuan and Wojciech Zaremba and Rowan Zellers and Chong Zhang and Marvin Zhang and Shengjia Zhao and Tianhao Zheng and Juntang Zhuang and William Zhuk and Barret Zoph},
      year={2024},
      eprint={2303.08774},
      archivePrefix={arXiv},
      primaryClass={cs.CL},
      url={https://arxiv.org/abs/2303.08774}, 
}

@article{touvron2023llama,
  title={Llama: Open and efficient foundation language models},
  author={Touvron, Hugo and Lavril, Thibaut and Izacard, Gautier and Martinet, Xavier and Lachaux, Marie-Anne and Lacroix, Timoth{\'e}e and Rozi{\`e}re, Baptiste and Goyal, Naman and Hambro, Eric and Azhar, Faisal and others},
  journal={arXiv preprint arXiv:2302.13971},
  year={2023}
}

@article{anthropic2024claude,
  title={The claude 3 model family: Opus, sonnet, haiku},
  author={Anthropic, AI},
  journal={Claude-3 Model Card},
  volume={1},
  number={1},
  pages={4},
  year={2024}
}

@inproceedings{liu2025palimpzest,
  title={Palimpzest: Optimizing ai-powered analytics with declarative query processing},
  author={Liu, Chunwei and Russo, Matthew and Cafarella, Michael and Cao, Lei and Chen, Peter Baile and Chen, Zui and Franklin, Michael and Kraska, Tim and Madden, Samuel and Shahout, Rana and others},
  booktitle={Proceedings of the Conference on Innovative Database Research (CIDR)},
  pages={2},
  year={2025}
}

@misc{patel2024semanticoperators,
      title={Semantic Operators: A Declarative Model for Rich, AI-based Analytics Over Text Data},
      author={Liana Patel and Siddharth Jha and Parth Asawa and Melissa Pan and Carlos Guestrin and Matei Zaharia},
      year={2024},
      eprint={2407.11418},
      archivePrefix={arXiv},
      primaryClass={cs.DB},
      url={https://arxiv.org/abs/2407.11418},
}

@article{10.1145/235968.233342,
author = {Poosala, Viswanath and Haas, Peter J. and Ioannidis, Yannis E. and Shekita, Eugene J.},
title = {Improved histograms for selectivity estimation of range predicates},
year = {1996},
issue_date = {June 1996},
publisher = {Association for Computing Machinery},
address = {New York, NY, USA},
volume = {25},
number = {2},
issn = {0163-5808},
url = {https://doi.org/10.1145/235968.233342},
doi = {10.1145/235968.233342},
journal = {SIGMOD Rec.},
month = jun,
pages = {294–305},
numpages = {12}
}

@inproceedings{10.1145/3626246.3653368,
author = {Lamb, Andrew and Shen, Yijie and Heres, Dani\"{e}l and Chakraborty, Jayjeet and Kabak, Mehmet Ozan and Hsieh, Liang-Chi and Sun, Chao},
title = {Apache Arrow DataFusion: A Fast, Embeddable, Modular Analytic Query Engine},
year = {2024},
isbn = {9798400704222},
publisher = {Association for Computing Machinery},
address = {New York, NY, USA},
url = {https://doi.org/10.1145/3626246.3653368},
doi = {10.1145/3626246.3653368},
booktitle = {Companion of the 2024 International Conference on Management of Data},
pages = {5–17},
numpages = {13},
location = {Santiago AA, Chile},
series = {SIGMOD '24}
}

@online{apache_arrow,
  author  = {{Apache Software Foundation}},
  title   = {{Apache Arrow}},
  year    = {2026},
  url     = {https://arrow.apache.org/},
  urldate = {2026-03-12}
}

@article{10.14778/2850583.2850594,
author = {Leis, Viktor and Gubichev, Andrey and Mirchev, Atanas and Boncz, Peter and Kemper, Alfons and Neumann, Thomas},
title = {How good are query optimizers, really?},
year = {2015},
issue_date = {November 2015},
publisher = {VLDB Endowment},
volume = {9},
number = {3},
issn = {2150-8097},
url = {https://doi.org/10.14778/2850583.2850594},
doi = {10.14778/2850583.2850594},
journal = {Proc. VLDB Endow.},
month = nov,
pages = {204–215},
numpages = {12}
}

@inproceedings{10.1145/3600006.3613165,
author = {Kwon, Woosuk and Li, Zhuohan and Zhuang, Siyuan and Sheng, Ying and Zheng, Lianmin and Yu, Cody Hao and Gonzalez, Joseph and Zhang, Hao and Stoica, Ion},
title = {Efficient Memory Management for Large Language Model Serving with PagedAttention},
year = {2023},
isbn = {9798400702297},
publisher = {Association for Computing Machinery},
address = {New York, NY, USA},
url = {https://doi.org/10.1145/3600006.3613165},
doi = {10.1145/3600006.3613165},
booktitle = {Proceedings of the 29th Symposium on Operating Systems Principles},
pages = {611–626},
numpages = {16},
location = {Koblenz, Germany},
series = {SOSP '23}
}

@article{10.14778/3503585.3503586,
author = {Han, Yuxing and Wu, Ziniu and Wu, Peizhi and Zhu, Rong and Yang, Jingyi and Tan, Liang Wei and Zeng, Kai and Cong, Gao and Qin, Yanzhao and Pfadler, Andreas and Qian, Zhengping and Zhou, Jingren and Li, Jiangneng and Cui, Bin},
title = {Cardinality estimation in DBMS: a comprehensive benchmark evaluation},
year = {2021},
issue_date = {December 2021},
publisher = {VLDB Endowment},
volume = {15},
number = {4},
issn = {2150-8097},
url = {https://doi.org/10.14778/3503585.3503586},
doi = {10.14778/3503585.3503586},
journal = {Proc. VLDB Endow.},
month = {dec},
pages = {752–765},
numpages = {14}
}

@article{eleet1,
  title={Eleet: Efficient learned query execution over text and tables},
  author={Urban, Matthias and Binnig, Carsten},
  journal={Proceedings of the VLDB Endowment},
  volume={17},
  number={13},
  pages={4867--4880},
  year={2024},
  publisher={VLDB Endowment}
}

@article{eleet2,
  title={Efficient learned query execution over text and tables [technical report]},
  author={Urban, Matthias and Binnig, Carsten},
  journal={arXiv preprint arXiv:2410.22522},
  year={2024}
}

@article{lotus,
  title={Semantic Operators: A Declarative Model for Rich, AI-based Data Processing},
  author={Patel, Liana and Jha, Siddharth and Pan, Melissa and Gupta, Harshit and Asawa, Parth and Guestrin, Carlos and Zaharia, Matei},
  journal={arXiv preprint arXiv:2407.11418},
  year={2024}
}

@article{approximate_selection_with_guarantees,
  title={Approximate selection with guarantees using proxies},
  author={Kang, Daniel and Gan, Edward and Bailis, Peter and Hashimoto, Tatsunori and Zaharia, Matei},
  journal={arXiv preprint arXiv:2004.00827},
  year={2020}
}

@article{palimpzest,
  title={A declarative system for optimizing ai workloads},
  author={Liu, Chunwei and Russo, Matthew and Cafarella, Michael and Cao, Lei and Chen, Peter Baille and Chen, Zui and Franklin, Michael and Kraska, Tim and Madden, Samuel and Vitagliano, Gerardo},
  journal={arXiv preprint arXiv:2405.14696},
  year={2024}
}

@article{blendsql,
  title={Blendsql: A scalable dialect for unifying hybrid question answering in relational algebra},
  author={Glenn, Parker and Dakle, Parag Pravin and Wang, Liang and Raghavan, Preethi},
  journal={arXiv preprint arXiv:2402.17882},
  year={2024}
}

@article{tag,
  title={Text2sql is not enough: Unifying ai and databases with tag},
  author={Biswal, Asim and Patel, Liana and Jha, Siddarth and Kamsetty, Amog and Liu, Shu and Gonzalez, Joseph E and Guestrin, Carlos and Zaharia, Matei},
  journal={arXiv preprint arXiv:2408.14717},
  year={2024}
}

@article{abacus,
  title={Abacus: A Cost-Based Optimizer for Semantic Operator Systems},
  author={Russo, Matthew and Sudhir, Sivaprasad and Vitagliano, Gerardo and Liu, Chunwei and Kraska, Tim and Madden, Samuel and Cafarella, Michael},
  journal={arXiv preprint arXiv:2505.14661},
  year={2025}
}

@article{Galois,
  title={Querying large language models with SQL},
  author={Saeed, Mohammed and De Cao, Nicola and Papotti, Paolo},
  journal={arXiv preprint arXiv:2304.00472},
  year={2023}
}

@article{zendb,
  title={Towards accurate and efficient document analytics with large language models},
  author={Lin, Yiming and Hulsebos, Madelon and Ma, Ruiying and Shankar, Shreya and Zeigham, Sepanta and Parameswaran, Aditya G and Wu, Eugene},
  journal={arXiv preprint arXiv:2405.04674},
  year={2024}
}

@inproceedings{imdb,
  title={Learning word vectors for sentiment analysis},
  author={Maas, Andrew and Daly, Raymond E and Pham, Peter T and Huang, Dan and Ng, Andrew Y and Potts, Christopher},
  booktitle={Proceedings of the 49th annual meeting of the association for computational linguistics: Human language technologies},
  pages={142--150},
  year={2011}
}

@inproceedings{wannadb1,
  title={WannaDB: Ad-hoc Structured Exploration of Text Collections Using Queries.},
  author={H{\"a}ttasch, Benjamin},
  booktitle={DESIRES},
  pages={179--180},
  year={2021}
}

@inproceedings{wannadb2,
  title={WannaDB: Ad-hoc SQL Queries over Text Collections},
  author={H{\"a}ttasch, Benjamin and Bodensohn, Jan-Micha and Vogel, Liane and Urban, Matthias and Binnig, Carsten},
  booktitle={BTW 2023},
  pages={157--181},
  year={2023},
  organization={Gesellschaft f{\"u}r Informatik eV}
}

@article{deepscholar,
  title={DeepScholar-Bench: A Live Benchmark and Automated Evaluation for Generative Research Synthesis},
  author={Patel, Liana and Arabzadeh, Negar and Gupta, Harshit and Sundar, Ankita and Stoica, Ion and Zaharia, Matei and Guestrin, Carlos},
  journal={arXiv preprint arXiv:2508.20033},
  year={2025}
}

@inproceedings{thalamusdb1,
  title={Demonstration of thalamusdb: answering complex sql queries with natural language predicates on multi-modal data},
  author={Jo, Saehan and Trummer, Immanuel},
  booktitle={Companion of the 2023 International Conference on Management of Data},
  pages={179--182},
  year={2023}
}

@article{thalamusdb2,
  title={Thalamusdb: Approximate query processing on multi-modal data},
  author={Jo, Saehan and Trummer, Immanuel},
  journal={Proceedings of the ACM on Management of Data},
  volume={2},
  number={3},
  pages={1--26},
  year={2024},
  publisher={ACM New York, NY, USA}
}

@article{uqe,
  title={UQE: A Query Engine for Unstructured Databases},
  author={Dai, Hanjun and Wang, Bethany and Wan, Xingchen and Dai, Bo and Yang, Sherry and Nova, Azade and Yin, Pengcheng and Phothilimthana, Mangpo and Sutton, Charles and Schuurmans, Dale},
  journal={Advances in Neural Information Processing Systems},
  volume={37},
  pages={29807--29838},
  year={2024}
}

@article{xmode,
  title={Explainable Multi-Modal Data Exploration in Natural Language via LLM Agent},
  author={Nooralahzadeh, Farhad and Zhang, Yi and Furst, Jonathan and Stockinger, Kurt},
  journal={arXiv preprint arXiv:2412.18428},
  year={2024}
}

@article{docetl,
  title={Docetl: Agentic query rewriting and evaluation for complex document processing},
  author={Shankar, Shreya and Chambers, Tristan and Shah, Tarak and Parameswaran, Aditya G and Wu, Eugene},
  journal={arXiv preprint arXiv:2410.12189},
  year={2024}
}

@article{swan,
  title={Hybrid querying over relational databases and large language models},
  author={Zhao, Fuheng and Agrawal, Divyakant and Abbadi, Amr El},
  journal={arXiv preprint arXiv:2408.00884},
  year={2024}
}

@article{andb,
  title={AnDB: Breaking Boundaries with an AI-Native Database for Universal Semantic Analysis},
  author={Wang, Tianqing and Xue, Xun and Li, Guoliang and Wang, Yong},
  journal={arXiv preprint arXiv:2502.13805},
  year={2025}
}

@article{agenticdata,
  title={Agenticdata: An agentic data analytics system for heterogeneous data},
  author={Sun, Ji and Li, Guoliang and Zhou, Peiyao and Ma, Yihui and Xu, Jingzhe and Li, Yuan},
  journal={arXiv preprint arXiv:2508.05002},
  year={2025}
}

@article{aryn,
  title={The design of an llm-powered unstructured analytics system},
  author={Anderson, Eric and Fritz, Jonathan and Lee, Austin and Li, Bohou and Lindblad, Mark and Lindeman, Henry and Meyer, Alex and Parmar, Parth and Ranade, Tanvi and Shah, Mehul A and others},
  journal={arXiv preprint arXiv:2409.00847},
  year={2024}
}

@article{docdb,
  title={DocDB: A Database for Unstructured Document Analysis},
  author={Li, Zequn and Zhong, Yuanhao and Chai, Chengliang and Sun, Zhaoze and Deng, Yuhao and Yuan, Ye and Wang, Guoren and Cao, Lei},
  journal={Proceedings of the VLDB Endowment},
  volume={18},
  number={12},
  pages={5387--5390},
  year={2025},
  publisher={VLDB Endowment}
}

@inproceedings{doctopus,
  title={Doctopus: A System for Budget-aware Structural Data Extraction from Unstructured Documents},
  author={Zhong, Yuanhao and Deng, Yuhao and Chai, Chengliang and Gu, Ruixin and Yuan, Ye and Wang, Guoren and Cao, Lei},
  booktitle={Companion of the 2025 International Conference on Management of Data},
  pages={275--278},
  year={2025}
}

@inproceedings{exasol,
  title={Exasol Text AI: In-Database Continuous Text Analytics},
  author={Beck, Mikhail and Knapp Corneli, Madeleine and Kilias, Torsten and Kre{\ss}, Marlene and Kuhnke, Christoph and N{\"a}tlitz, Marco and {\"U}bensee, Thomas},
  booktitle={INFORMATIK 2025},
  pages={1509},
  year={2025},
  organization={Gesellschaft f{\"u}r Informatik eV}
}

@article{swelldb,
  title={SwellDB: GenAI-Native Query Processing via On-the-Fly Table Generation},
  author={Giannakouris, Victor},
  journal={Proceedings of the VLDB Endowment. ISSN},
  volume={2150},
  pages={8097}
}

@article{unify,
  title={Unify: A System For Unstructured Data Analytics},
  author={Wang, Jiayi and Li, Yuan and Wu, Jianming and Xu, Shihui and Li, Guoliang},
  journal={Proceedings of the VLDB Endowment},
  volume={18},
  number={12},
  pages={5287--5290},
  year={2025},
  publisher={VLDB Endowment}
}

@article{tabagent,
  title={TabAgent: A Multi-Agent Table Extraction Framework for Unstructured Documents},
  author={Wu, Jingfei and Shen, Chaoyuan and Deng, Qiyan and Wang, Yuping and Li, Jiajun and Deng, Yuhao and Yu, Minghe},
  journal={Proceedings of the VLDB Endowment. ISSN},
  volume={2150},
  pages={8097}
}

@article{weaver,
  title={Weaver: Interweaving SQL and LLM for Table Reasoning},
  author={Khoja, Rohit and Gupta, Devanshu and Fu, Yanjie and Roth, Dan and Gupta, Vivek},
  journal={arXiv preprint arXiv:2505.18961},
  year={2025}
}

@article{scaledoc,
  title={ScaleDoc: Scaling LLM-based Predicates over Large Document Collections},
  author={Zhang, Hengrui and Hui, Yulong and Liu, Yihao and Zhang, Huanchen},
  journal={arXiv preprint arXiv:2509.12610},
  year={2025}
}

@misc{sembench,
      title={SemBench: A Benchmark for Semantic Query Processing Engines}, 
      author={Jiale Lao and Andreas Zimmerer and Olga Ovcharenko and Tianji Cong and Matthew Russo and Gerardo Vitagliano and Michael Cochez and Fatma Özcan and Gautam Gupta and Thibaud Hottelier and H. V. Jagadish and Kris Kissel and Sebastian Schelter and Andreas Kipf and Immanuel Trummer},
      year={2025},
      eprint={2511.01716},
      archivePrefix={arXiv},
      primaryClass={cs.DB},
      url={https://arxiv.org/abs/2511.01716}, 
}

@article{doi:10.1137/0117039,
author = {Graham, R. L.},
title = {Bounds on Multiprocessing Timing Anomalies},
journal = {SIAM Journal on Applied Mathematics},
volume = {17},
number = {2},
pages = {416-429},
year = {1969},
doi = {10.1137/0117039},

URL = { 
    
        https://doi.org/10.1137/0117039
    
    

},
eprint = { 
    
        https://doi.org/10.1137/0117039
    
    

}

}

@article{singh2025openai,
  title={Openai gpt-5 system card},
  author={Singh, Aaditya and Fry, Adam and Perelman, Adam and Tart, Adam and Ganesh, Adi and El-Kishky, Ahmed and McLaughlin, Aidan and Low, Aiden and Ostrow, AJ and Ananthram, Akhila and others},
  journal={arXiv preprint arXiv:2601.03267},
  year={2025}
}

@article{yang2025qwen3,
  title={Qwen3 technical report},
  author={Yang, An and Li, Anfeng and Yang, Baosong and Zhang, Beichen and Hui, Binyuan and Zheng, Bo and Yu, Bowen and Gao, Chang and Huang, Chengen and Lv, Chenxu and others},
  journal={arXiv preprint arXiv:2505.09388},
  year={2025}
}

@inproceedings{maas-etal-2011-learning,
    title = "Learning Word Vectors for Sentiment Analysis",
    author = "Maas, Andrew L.  and
      Daly, Raymond E.  and
      Pham, Peter T.  and
      Huang, Dan  and
      Ng, Andrew Y.  and
      Potts, Christopher",
    editor = "Lin, Dekang  and
      Matsumoto, Yuji  and
      Mihalcea, Rada",
    booktitle = "Proceedings of the 49th Annual Meeting of the Association for Computational Linguistics: Human Language Technologies",
    month = jun,
    year = "2011",
    address = "Portland, Oregon, USA",
    publisher = "Association for Computational Linguistics",
    url = "https://aclanthology.org/P11-1015/",
    pages = "142--150"
}

@inproceedings{10.5555/646975.711521,
author = {Raab, Martin and Steger, Angelika},
title = {"Balls into Bins" - A Simple and Tight Analysis},
year = {1998},
isbn = {354065142X},
publisher = {Springer-Verlag},
address = {Berlin, Heidelberg},
booktitle = {Proceedings of the Second International Workshop on Randomization and Approximation Techniques in Computer Science},
pages = {159–170},
numpages = {12},
series = {RANDOM '98}
}

@article{10.1145/3769776,
author = {Zeighami, Sepanta and Shankar, Shreya and Parameswaran, Aditya},
title = {Cut Costs, Not Accuracy: LLM-Powered Data Processing with Guarantees},
year = {2025},
issue_date = {December 2025},
publisher = {Association for Computing Machinery},
address = {New York, NY, USA},
volume = {3},
number = {6},
url = {https://doi.org/10.1145/3769776},
doi = {10.1145/3769776},
journal = {Proc. ACM Manag. Data},
month = dec,
articleno = {311},
numpages = {26}
}

@inproceedings{lin2024batchprompt,
  title={Batchprompt: Accomplish more with less},
  author={Lin, Jianzhe and Diesendruck, Maurice and Du, Liang and Abraham, Robin},
  booktitle={International Conference on Learning Representations},
  volume={2024},
  pages={21590--21612},
  year={2024}
}

@inproceedings{Cheng2023BatchPE,
  title={Batch Prompting: Efficient Inference with Large Language Model APIs},
  author={Zhoujun Cheng and Jungo Kasai and Tao Yu},
  booktitle={Conference on Empirical Methods in Natural Language Processing},
  year={2023},
  url={https://api.semanticscholar.org/CorpusID:256080546}
}

@inproceedings{10.1145/3696410.3714950,
author = {Cong, Peizhuang and Chen, Qizhi and Zhao, Haochen and Yang, Tong},
title = {BATON: Enhancing Batch-wise Inference Efficiency for Large Language Models via Dynamic Re-batching},
year = {2025},
isbn = {9798400712746},
publisher = {Association for Computing Machinery},
address = {New York, NY, USA},
url = {https://doi.org/10.1145/3696410.3714950},
doi = {10.1145/3696410.3714950},
booktitle = {Proceedings of the ACM on Web Conference 2025},
pages = {2309–2318},
numpages = {10},
location = {Sydney NSW, Australia},
series = {WWW '25}
}

@article{10.14778/3734839.3734853,
author = {Ji, Zhaoxuan and Wang, Xinlu and Luo, Zhaojing and Xie, Zhongle and Zhang, Meihui},
title = {Optimized Batch Prompting for Cost-Effective LLMs},
year = {2025},
issue_date = {March 2025},
publisher = {VLDB Endowment},
volume = {18},
number = {7},
issn = {2150-8097},
url = {https://doi.org/10.14778/3734839.3734853},
doi = {10.14778/3734839.3734853},
journal = {Proc. VLDB Endow.},
month = mar,
pages = {2172–2184},
numpages = {13}
}

@inproceedings{Liu2024CliqueParcelAA,
  title={CliqueParcel: An Approach For Batching LLM Prompts That Jointly Optimizes Efficiency And Faithfulness},
  author={Jia-Yi Liu and Tinghan Yang and Jennifer Neville},
  booktitle={SmaLLEXT@CIKM},
  year={2024},
  url={https://api.semanticscholar.org/CorpusID:267897588}
}

@article{Xu2026TowardsCL,
  title={Towards Cost-effective LLMs Routing with Batch Prompting},
  author={Haotian Xu and Kangfei Zhao and Jiadong Xie},
  journal={ArXiv},
  year={2026},
  volume={abs/2605.28268},
  url={https://api.semanticscholar.org/CorpusID:288701848}
}

@article{qi_vldb26,
        title = {{Sema: A High-performance System for LLM-based Semantic Query Processing}},
        author = {Qi, Kangkang and Xie, Dongyang and Li, Wenbo and Zhang, Hao and Zhu, Yuanyuan and Yu, Jeffrey Xu and Zhao, Kangfei},
        journal = {PVLDB},
        series = {{VLDB} '26},
        volume = {19},
        number = {11},
        pages = {3231--3244},
        doi = {10.14778/3836663.3836685},
        url = {https://doi.org/10.14778/3836663.3836685},
        year = {2026}
}

@article{10.1093/biomet/26.4.404,
    author = {CLOPPER, C. J. and PEARSON, E. S.},
    title = {THE USE OF CONFIDENCE OR FIDUCIAL LIMITS ILLUSTRATED IN THE CASE OF THE BINOMIAL},
    journal = {Biometrika},
    volume = {26},
    number = {4},
    pages = {404-413},
    year = {1934},
    month = {12},
    issn = {0006-3444},
    doi = {10.1093/biomet/26.4.404},
    url = {https://doi.org/10.1093/biomet/26.4.404},
    eprint = {https://academic.oup.com/biomet/article-pdf/26/4/404/823407/26-4-404.pdf},
}

\end{document}